\pdfoutput=1

\documentclass[sigconf]{acmart}

\setcopyright{none}

\acmConference[Conference acronym 'XX]{Make sure to enter the correct
  conference title from your rights confirmation emai}{June 03--05,
  2018}{Woodstock, NY}
\acmPrice{15.00}
\acmISBN{978-1-4503-XXXX-X/18/06}

\copyrightyear{2023} 
\acmYear{2023} 
\setcopyright{acmlicensed}\acmConference[KDD '23]{Proceedings of the 29th ACM SIGKDD Conference on Knowledge Discovery and Data Mining}{August 6--10, 2023}{Long Beach, CA, USA}
\acmBooktitle{Proceedings of the 29th ACM SIGKDD Conference on Knowledge Discovery and Data Mining (KDD '23), August 6--10, 2023, Long Beach, CA, USA}
\acmPrice{15.00}
\acmDOI{10.1145/3580305.3599776}
\acmISBN{979-8-4007-0103-0/23/08}

\usepackage{balance}
\usepackage{amsmath}
\usepackage{url}
\usepackage{makecell}
\usepackage{scalerel}
\usepackage{balance}
\usepackage{epstopdf}
\usepackage{graphicx}
\usepackage{subfigure}
\usepackage{color}
\usepackage{tabularx}
\usepackage{hyperref}
\usepackage{ragged2e}  
\usepackage{booktabs}
\usepackage[noend]{algpseudocode}
\usepackage{fixltx2e}
\usepackage{paralist}
\usepackage{listings}
\usepackage{multirow}
\usepackage{enumerate}
\usepackage{bbm}
\usepackage{comment}
\usepackage{caption} 
\usepackage{stfloats}
\usepackage{bm}
\usepackage{multicol}
\usepackage{comment}
\usepackage[labelfont=bf]{caption}
\usepackage{enumitem}
\usepackage{etoolbox} 
\usepackage{mdwlist} 

\usepackage[noend,lined,boxed,vlined,ruled,linesnumbered]{algorithm2e}
\SetKwInOut{Input}{input}
\SetKwInOut{Output}{output}

\newcommand{\val}[1]{{\small \texttt{#1}}}

\newcommand{\yeye}[1]{\textcolor{blue}{Yeye: #1}}
\newcommand{\avhfull}{\textsc{Auto-Validate-by-History}\xspace}
\newcommand{\avh}{\textsc{AVH}\xspace}

\newtoggle{fullversion}
\toggletrue{fullversion}

\newenvironment{proofsketch}{\paragraph{\textsc{\normalfont{Proof Sketch}}:}}{\hfill$\square$}

\newcounter{definition}
\newenvironment{definition}[1][]{\refstepcounter{definition}\par\smallskip\textsc{Definition~\thedefinition.\ #1}}{\smallskip}

\newcounter{example}
\newenvironment{example}[1][]{\refstepcounter{example}\par\smallskip\textsc{Example~\theexample.\ #1}}{\smallskip}

\newcounter{lemma}

\newcounter{proposition}
\newenvironment{proposition}[1][]{\refstepcounter{proposition}\par\smallskip\textsc{Proposition~\theproposition.\ #1} \itshape}{\smallskip}

\begin{document}

\title{Auto-Validate by-History: Auto-Program Data Quality Constraints to Validate Recurring Data Pipelines}



\author{Dezhan Tu}
\authornote{Work done at Microsoft.}
\affiliation{%
  \institution{University of California, Los Angeles}
}

\author{Yeye He, Weiwei Cui, Song Ge, Haidong Zhang, Han Shi, Dongmei Zhang, Surajit Chaudhuri}
\affiliation{%
  \institution{Microsoft Research}
}

\renewcommand{\authors}{Dezhan Tu, Yeye He, Weiwei Cui, Song Ge, Haidong Zhang, Han Shi, Dongmei Zhang, Surajit Chaudhuri}


\begin{abstract}
Data pipelines are widely employed in modern enterprises to power a variety of Machine-Learning (ML) 
and Business-Intelligence (BI) applications. Crucially, these pipelines are \emph{recurring} (e.g., daily or hourly) in production settings to keep data updated so that ML models can be re-trained regularly, and BI dashboards
refreshed frequently. 
However, data quality (DQ) issues can often creep into recurring pipelines because of upstream schema and data drift over time. As modern enterprises operate thousands of recurring pipelines, today data engineers have to spend substantial efforts to \emph{manually} monitor and resolve DQ issues, as part of their DataOps and MLOps practices.

Given the high human cost of managing large-scale pipeline operations, it is imperative that we can \emph{automate} as much as possible.
In this work, we propose \avhfull (\avh) that can automatically detect DQ issues in recurring pipelines, leveraging rich statistics from historical executions. 
We formalize this as an optimization problem, and develop constant-factor approximation algorithms with provable precision guarantees.
Extensive evaluations using 2000
production data pipelines at Microsoft demonstrate
the effectiveness and efficiency of \avh.

\end{abstract}


\maketitle
\pagestyle{plain}

\section{Introduction}
Data pipelines are the crucial infrastructure underpinning the modern data-driven economy. 
Today, data pipelines are ubiquitous in large technology companies such as Amazon, Google and Microsoft to power data-hungry businesses like search
and advertisement~\cite{schelter2018deequ, breck2019data, song2021auto, zhou2012scope}. Pipelines are also increasingly used in traditional enterprises across a variety of ML/BI applications, in a growing trend to democratize data~\cite{gartner}.

Production data pipelines are often \emph{inter-dependent}, forming complex ``webs'', where input tables used by downstream pipelines frequently depend on output tables from upstream pipelines.

Furthermore, these pipelines are often configured to \emph{recur} on a regular basis (e.g., hourly or daily), to ensure data stay up-to-date for downstream use cases (e.g., fresh data enables ML models to be re-trained regularly, and BI dashboards refreshed continuously).

\textbf{Recurring Pipelines: Prone to Fail due to DQ.} 
The \emph{recurring} and \emph{inter-dependent} nature of production pipelines make them vulnerable to failure due to data quality (DQ) issues,
because over time unexpected DQ issues, such
as data drift~\cite{polyzotis2017data} and schema drift~\cite{schelter2018automating, breck2019data},
can creep in, causing cascading issues in pipelines.

Although DQ issues in data pipelines are widely documented in the literature (especially in industry settings~\cite{schelter2018deequ, breck2019data, song2021auto, zhou2012scope}), we describe a few common types of DQ issues from the literature, in order to make the discussion concrete and self-contained: 
\begin{itemize}[leftmargin=*,noitemsep,topsep=0pt,parsep=0pt,partopsep=0pt]
\item \emph{\underline{Schema drift}}: A newly arrived batch of input data may have a schema change compared to previous input (e.g., missing columns or extra columns), which can result in incorrect behavior in data pipelines~\cite{schelter2018automating, breck2019data}.
\item \emph{\underline{Increasing nulls}}: There is sometimes a sudden increase of null, empty strings, or special values (e.g., -1) in a column due to external factors -- for instance, Google reports a DQ incident where null values in a column increase substantially in a short period of time, because the module that populates data in this column encountered an unusual number of RPC time-outs from a networking outage~\cite{polyzotis2017data}.
\item \emph{\underline{Change of units}}: The unit of measurement for numeric values can change over time, when the logic that populates data evolves -- for instance, Google reports a real DQ issue in their search ranking~\cite{polyzotis2017data}, where the program that populates the ``age'' field of web documents previously used the unit of ``days'' (e.g., a document that is 30 days old will have an ``age'' value of 30), which later got changed to  ``hours'' (making the same document to be have the ``age'' value of 720). This leads to orders of magnitude larger ``age'' values, and incorrect behaviours downstream. 
\item \emph{\underline{Change of value standards}}: Value standards for string-valued data can change over time -- for instance, Amazon reports a DQ issue where a ``\val{language-locale}'' column previously used lowercase values like ``\val{en-gb}'', which later changed into uppercase ``\val{en-GB}'', creating a mixed bag of inconsistent values in the same column, leading to incorrect behaviours in downstream applications~\cite{schelter2018automating}.
\item \emph{\underline{Change of data volume}}: The volume (e.g., row-count) for a new batch of data in a recurring pipeline can change significantly from previous batches, which can also be indicative of DQ issues. 
\end{itemize}
This list of DQ issues is clearly not exhaustive as there are many other types of DQ issues documented in the literature~\cite{schelter2018automating, schelter2018deequ, breck2019data}.

\begin{figure*}[t]
 \vspace{-14mm}
    \centering
    \hspace{-3mm}
    \subfigure[Amazon's Deequ: each arrow points to a column-level constraint ]{\label{fig:deequ}\includegraphics[height=2.7cm]{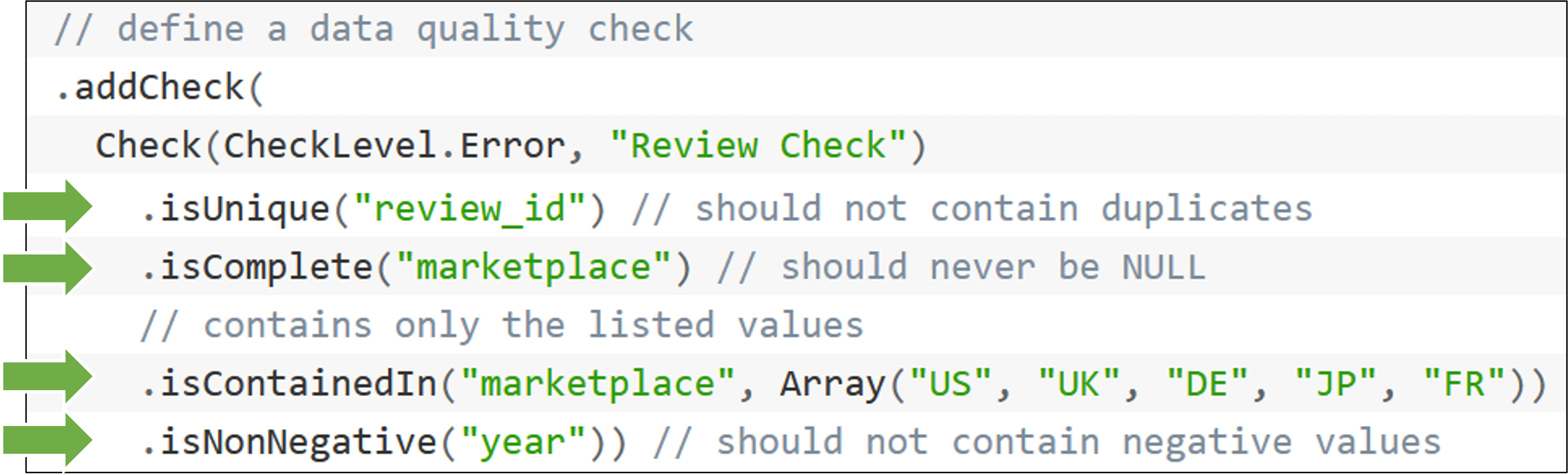}} 
    \hspace{1mm}
    \subfigure[Google's TFDV: each arrow points to a column-level constraint]{\label{fig:tfdv}\includegraphics[height=3cm]{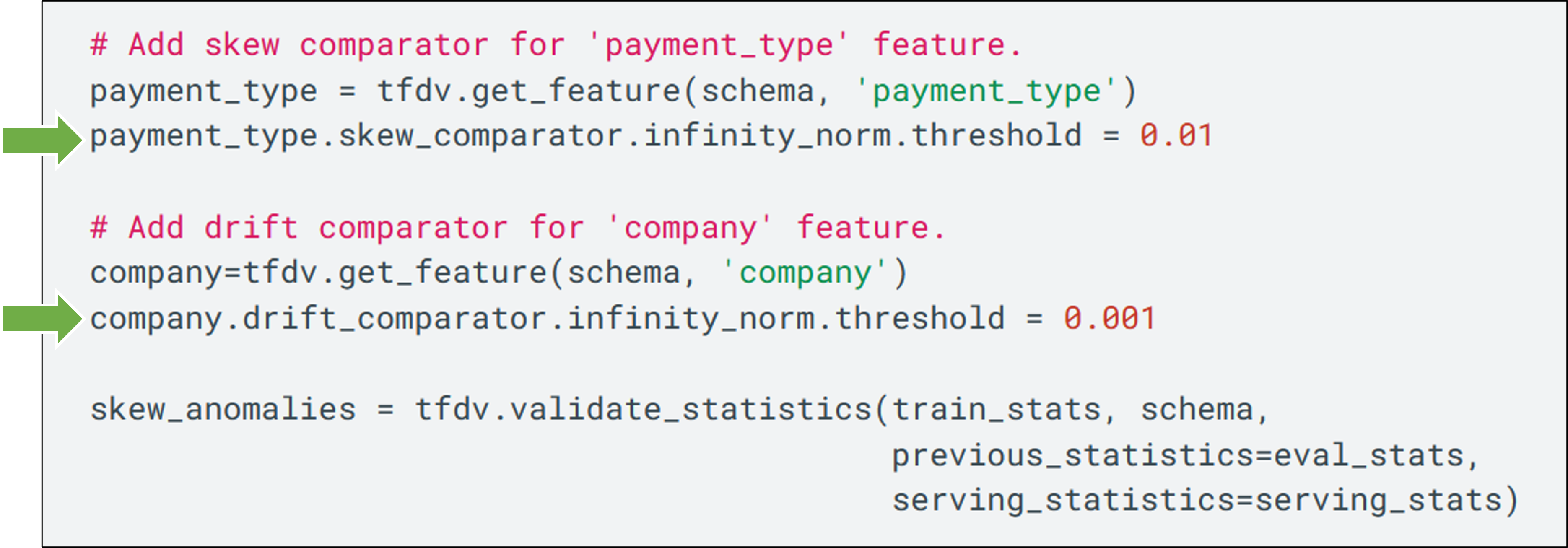}}
    \vspace{-5mm}
    \caption{Declarative Data-validation using manually-programmed constraints}
    \label{fig:DSL}
    \vspace{-5mm}
\end{figure*}

When DQ issues arise in recurring
pipelines, they tend to introduce \emph{silent} failures (i.e., with no explicit exceptions thrown, or error messages generated). The silent nature of DQ issues makes them difficult to catch, but no less damaging. For example, when null values increase significantly, or the unit of measurement changes, downstream ML models will continue to operate but will likely churn out inaccurate predictions (e.g., Google reports a DQ issue in their production pipelines that causes a recommendation model in Google Store to produce sub-optimal results -- fixing this single DQ issue improves their apps install rate by 2\%~\cite{breck2019data}). 

In general,  ``silent'' DQ failures pollute
downstream data products, which makes it more time-consuming for engineers to detect/debug/fix. Silent DQ failures in pipelines are therefore a major pain point in MLOps and DataOps practices~\cite{schelter2018automating, schelter2018deequ, breck2019data, song2021auto}.

\textbf{``Guardrails'' for Pipelines: Data Validation}. Technology companies with large-scale data pipeline operations are among the first to recognize the need for employing data-validation "guardrails" in recurring pipelines to catch DQ issues early as they arise. A number of data-validation tools have been developed,
including Google’s TensorFlow Data Validation 
(TFDV)~\cite{caveness2020tensorflow} and Amazon’s Deequ~\cite{schelter2018deequ, schelter2019unit}. 

These tools develop easy-to-use domain-specific languages (DSLs), so that 
engineers can write declarative DQ 
constraints that describe how “normal” data
should look like in recurring data pipelines, such that
unexpected deviation in the future can be flagged for review.

Figure~\ref{fig:deequ} shows an example code snippet 
from Amazon Deequ. Using the DSL introduced in Deequ, one could
declare the ``\val{review\_id}'' column to be unique, the 
``\val{marketplace}'' column to be complete (with no NULLs), 
etc.. These constraints are then used to validate future data arriving in the recurring pipeline. 

Figure~\ref{fig:tfdv} shows a similar example 
from Google’s TFDV, which specifies that, when a new batch of input data arrives in a pipeline,
the distributional distance of values in the ``\val{payment\_type}'' column should be similar to the same column from previous batches (the code snippet therefore specifies that the L-infinity distance of the two should be no greater 
than 0.01).


\textbf{Automate Data Validation: Leveraging History.}
While these DSL-based declarative data-validation solutions improve upon low-level assertions and can improve DQ in production pipelines as reported in~\cite{schelter2018automating, schelter2018deequ, breck2019data},
they require data engineers to manually write 
data constraints \emph{one-column-at-a-time} like shown in Figure~\ref{fig:DSL} (with a lone exception in~\cite{redyuk2021automating}, which employs off-the-shelf anomaly detection algorithms). This
is clearly time-consuming and hard to scale -- large organizations today operate thousands of pipelines, and hundreds of columns in each table, it is impractical for engineers to manually program DQ for each column.

We emphasize that writing DQ constraints is not just time-consuming, sometimes it is also genuinely difficult for humans to program DQ correctly, because users need to (1) have a deep understanding of the underlying data including how the data may evolve over time; and (2) be well-versed in complex statistical metrics (e.g., L-infinity vs. JS-divergence), before they can program DQ effectively. Consider the example of online user traffic, such data can fluctuate quickly over time (e.g., between different hours of the day and different days of the week), which is hard to anticipate and even harder to program using appropriate metrics and thresholds.

To address this common pain point, in this work
we propose to ``auto-program'' DQ, by leveraging ``history''. Our insight is that rich statistical information from past executions of the same pipeline (e.g., row-counts, unique-values, value-distributions, etc.) is readily available, which can serve as strong signals to reliably predict whether a new batch of data may have DQ issues or not. 

To see why this is the case, consider a simplistic example where all K past executions of a recurring pipeline produce exactly 50 output rows (one row for each of the 50 US states). This row-count becomes a ``statistical invariant'' unique to this particular pipeline, which can serve as a good predictor for DQ issues in the future -- deviations from the invariant in new executions (e.g. an output with only 10 rows or 0 row), would likely point to DQ issues. 

Obviously, simple row-counts are not the best DQ predictor for all pipelines, as some pipeline can have row counts that can vary significantly. In such cases, DQ constraints based on other types of statistical metrics will likely be more effective.

Table~\ref{tab:numeric} and Table~\ref{tab:category} list common statistical metrics used to program DQ constraints (details of these metrics can be found in Appendix~\ref{apx:detailed-statistics}). Tools like TFDV and Deequ already support many of these metrics today, but it is difficult for humans to manually select suitable metrics, and then guess what thresholds would work well. Our proposed \avh aims to automatically program suitable DQ from this large space of statistical primitives, so that the resulting DQ is tailored to the underlying pipeline, without human intervention.

Overall, our proposed \avh is designed to have the following properties that we believe are crucial in recurring pipelines:

\begin{itemize}[leftmargin=*, noitemsep,topsep=0pt,parsep=0pt,partopsep=0pt]
\item \underline{\textit{Automated}}. Instead of requiring humans to manually program DQ constraints column-at-a-time, \avh can auto-program rich DQ leveraging statistics from past executions.
\item \underline{\textit{Highly accurate}}. \avh is specifically designed to achieve high accuracy, as frequent false-alarms would require constant human attention that can erode user confidence quickly. In \avh, we aim for very low False-Positive-Rate or FPR (e.g., 0.01\%), which is also configurable if users so choose. \avh can then auto-program DQ guaranteed not to exceed that FPR, while still maximizing the expected ``recall'' (the number of DQ issues to catch). 
\item \underline{\textit{Robust}}. Unlike traditional ML methods that require a significant amount of training data, we exploit rich statistical properties (Chebyshev, Chantelli, CLT, etc.) of the underlying metrics, so that the predictions are robust even with limited historical data (e.g. only a few days of histories).
\item \underline{\textit{Explainable}}. \avh produces explainable DQ constraints using standard statistical metrics (as opposed to black-box models), which makes it possible for human engineers to understand, review, and approve, if human interactions become necessary. 
\end{itemize}

\noindent \textbf{Contributions.} We make the following contributions in this work.
\begin{itemize}[leftmargin=*, noitemsep,topsep=0pt,parsep=0pt,partopsep=0pt]
\item We propose a novel problem to auto-program pipeline DQ leveraging history, formalized as a principled optimization problem specifically optimizing both precision and recall.
\item We develop algorithms that leverage the different statistical properties of the underlying metrics, to achieve a constant-factor 
approximation, while having provable  precision guarantees.
\item Our extensive evaluation on 2000 real production pipelines suggests that \avh substantially outperforms a variety of commercial solutions, as well as  SOTA methods for anomaly detection from the machine learning and database literature. 
\end{itemize}


\vspace{-2mm}
\section{Related Work}
\textbf{Data validation.} Data validation for pipelines is an emerging topic that has attracted significant interest from the industry, including recent efforts such as Google’s TensorFlow Data Validation 
(TFDV)~\cite{caveness2020tensorflow}, Amazon’s Deequ~\cite{schelter2018deequ, schelter2019unit}, and
LinkedIn's Data Sentinel~\cite{swami2020data}. With the lone exception of~\cite{redyuk2021automating},
most existing work focuses on developing infrastructures and DSLs so that engineers can program DQ constraints in a declarative manner. 

\textbf{Anomaly detection.} Anomaly detection has been widely studied in time-series and tabular settings~\cite{ben2005outlier, aggarwal2001outlier, hodge2004survey, ferdousi2006unsupervised, ljung1993outlier, basu2007automatic}, and is clearly related.
We compare with an extensive set of over 10 SOTA methods from the anomaly detection literature, and show \avh is substantially better in our problem setting of DQ in pipelines, because \avh uniquely exploits the  statistical properties of underlying metrics, whereas standard anomaly detection methods would treat each metric as just another ``feature dimension''.  This enables \avh to have higher accuracy, and excel with even limited data (e.g., 7 days of historical data), as we will show in our experiments. 

\textbf{Data cleaning.} There is a large literature on data cleaning (e.g., surveyed in~\cite{rahm2000data, raman2001potter, galhardas2001declarative, hellerstein2008quantitative, li2021auto}), which we also compare with. Most existing work focuses on the static single-table setting~\cite{yan2020scoded, huang2018auto, wang2019uni, song2021auto, hellerstein2008quantitative, chu2013holistic, raman2001potter, chiang2008discovering}, where data errors need to be detected from one table snapshot. In comparison, we study how multiple historical table snapshots from recurring pipelines can be explicitly leveraged for data quality, which is a setting not traditionally considered in the data cleaning literature.
\vspace{-2mm}
\section{Preliminary: DQ in Pipelines}
\label{sec:preliminary}
In this section, we introduce necessary preliminaries for programming Data Quality (DQ) in the context of data pipelines.

\begin{small}
\begin{table*}[h]
\vspace{-14mm}
    \centering
    \begin{tabular}{cp{0.7\textwidth}}
    \toprule
    Type & Metrics  \\ 
    \midrule
    Two-distribution & Earth Mover's distance (EMD)~\cite{rubner1998metric},  Jensen–Shannon divergence ($JS\_div$)~\cite{endres2003new}, Kullback–Leibler divergence ($KL\_div$)~\cite{endres2003new}, Two-sample Kolmogorov–Smirnov test ($KS\_dist$)~\cite{massey1951kolmogorov}, Cohen's d ($Cohen\_d$)~  \cite{cohen2013statistical} \\
    \hline
    Single-distribution &  $min$, $max$, $mean$, $median$, $sum$, $range$, $row\_count$, $unique\_ratio$, $complete\_ratio$ \\
    \bottomrule 
    \end{tabular}
    \caption{Statistical metrics used to generate DQ constraints for numerical data (details in Appendix~\ref{apx:detailed-statistics}).}
    \label{tab:numeric}
\end{table*}
\end{small}

\begin{small}
\begin{table*}[!h]
\vspace{-8mm}
    \centering
    \begin{tabular}{cp{0.7\textwidth}}
    \toprule
     Type & Metrics  \\ 
    \midrule
    Two-distribution & L-1  distance\cite{cantrell2000modern}, L-infinity   distance\cite{cantrell2000modern}, Cosine distance\cite{gupta2021deep}, Chi-squared test\cite{fienberg1979use}, Jensen–Shannon divergence ($JS\_div$)~\cite{endres2003new}, Kullback–Leibler divergence ($KL\_div$)~\cite{endres2003new} \\
    \hline
    Single-distribution  & $str\_len$, $char\_len$, $digit\_len$, $punc\_len$, $row\_count$, $unique\_ratio$, $complete\_ratio$, $dist\_val\_count$\\
    \bottomrule 
    \end{tabular}
    \caption{Statistical metrics used to generate DQ constraints for categorical data (details in Appendix~\ref{apx:detailed-statistics}).}
    \vspace{-8mm}
    \label{tab:category}
\end{table*}
\end{small}

As we discussed, data pipelines are ubiquitous today, yet DQ issues are common in recurring pipelines, giving rise to data-validation tools such as Google’s TFDV and  Amazon’s Deequ. 
At its core, these methods validate DQ by checking input/output tables of recurring pipelines, against pre-specified DQ constraints. Most of these constraints are defined over a single column $C$ at a time (Figure~\ref{fig:DSL}). Single-column DQ is also the type of DQ we focus in this work.

While TFDV and Deequ use syntactically different DSLs to program DQ constraints, the two are very similar in essence, as both can be described as constraints based on statistical metrics.

\textbf{DQ constraints by statistical metrics.}
Most DQ primitives used for pipeline validation can be expressed as statistics-based constraints. 
Table~\ref{tab:numeric} and Table~\ref{tab:category} list common statistical metrics used in DQ (e.g., row-count, L-infinity, etc.). We denote this space of possible metrics by $\mathbf{M}$. This is obviously a large space that requires time and expertise from users to navigate and select appropriately.

We now define two types of DQ constraints using metric $M \in \mathbf{M}$, which we call \emph{single-distribution} and \emph{two-distributions DQ constraints}, respectively.

\begin{definition}
\label{def:single}
A \emph{single-distribution DQ constraint}, denoted by a quadruple $Q(M, C, \theta_l, \theta_u)$, is defined using a statistical metric $M \in \mathbf{M}$, over a target column of data $C$, with lower-bound threshold $\theta_l$ and upper-bound $\theta_u$. The constraint $Q(M, C, \theta_l, \theta_u)$ specifies that $C$ has to satisfy the inequality $\theta_l \leq M(C) \leq \theta_u$, or the metric $M(C)$ is expected to be between the range $[\theta_l, \theta_u]$ (otherwise the constraint Q is deemed as violated).
\end{definition}

Single-distribution DQ can be instantiated using example metrics shown in Table~\ref{tab:numeric}
and Table~\ref{tab:category}. Such DQ constraints rely on a single data distribution of a column $C$, and can be validated using a newly-arrived batch data alone. We illustrate this using an example below.

\begin{example}
In the example Deequ snippet shown in Figure~\ref{fig:deequ}, the ``\val{review\_id}'' column is required to be unique, which can be expressed as a single-distribution DQ using the $unique\_ratio$ metric from Table~\ref{tab:category}, as $Q_1(unique\_ratio,$ $review\_id,$ $1, 1)$, equivalent to 
$Q_1: 1 \leq$ unique\_ratio(\val{review\_id}) $\leq 1$ 
(where the upper-bound $\theta_u$ and lower-bound $\theta_l$ converge to the same value 1). 
If on the other hand, uniqueness is required to be high, say at least 95\% (but not 100\%), we can write as $Q_2(unique\_ratio,$ $review\_id,$ $0.95, 1)$, or $Q_2: 0.95 \leq$ unique\_ratio(\val{review\_id}) $\leq 1$.
Other examples in Figure~\ref{fig:deequ} can be written as single-distribution DQ similarly.
\end{example}

Next, we introduce \emph{two-distribution DQ constraints} that require comparisons between two distributions of a column. 

\begin{definition}
\label{def:two}
A \emph{two-distribution DQ constraint}, denoted as $Q(M,$ $C, C', \theta_l, \theta_u)$, is defined using a statistical metric $M \in \mathbf{M}$, that compares one batch ``target'' data in a column $C$, and a batch of ``baseline'' data $C'$,  using lower-bound threshold $\theta_l$ and upper-bound threshold $\theta_u$. Formally we write $Q(M, C, C', \theta_l, \theta_u) = $ $\theta_l \leq M(C, C') \leq \theta_u$, which states that the metric $M(C, C')$ comparing $C$ and $C'$ is expected to be in the range $[\theta_l, \theta_u]$  (or Q is as violated otherwise).
\end{definition}

Two-distribution DQ compares a target column against a baseline column, which can be the same column from two consecutive executions of the same pipeline, or two batches of training/testing data, etc. We illustrate this in the example below.

\begin{example}
\label{ex:tfdv}
In the example TFDV snippet shown in Figure~\ref{fig:tfdv}, the first constraint specifies that for the ``\val{payment\_type}'' column, we expect two batches of the same data to differ by at most 0.01 using the L-infinity metric. This can be written as $Q_3: 0 \leq L_{inf}(C, C') \leq 0.01$.
The second constraint in Figure~\ref{fig:tfdv} defined on the ``\val{company}'' column, can be specified using two-distribution DQ similarly.
\end{example}

\textbf{Conjunctive DQ program.}
Given a target column $C$, it is often necessary to validate $C$ using multiple orthogonal metrics in $\mathbf{M}$ (e.g., both row-counts and distribution-similarity need to be checked, among other things). In this work, we consider conjunctions of multiple DQ constraints, which we call a \emph{conjunctive DQ program}.  Note that the use of conjunction is intuitive, as we want all DQ to hold at the same time
(prior work in TFDV and Deequ also implicitly employ conjunctions, as the example in Figure~\ref{fig:DSL} shows).

\begin{definition}
A \emph{conjunctive DQ program}, defined over a given set of (single-distribution or two-distribution) DQ constraints $\mathbf{S}$, denoted by $P(\mathbf{S})$, is defined as the conjunction of all $Q_i \in \mathbf{S}$, written as $P(\mathbf{S}) = \bigwedge_{Q_i \in \mathbf{S}} Q_i$.
\end{definition}

\begin{example}
Continue with Example~\ref{ex:tfdv}, let $C$ denotes the target column ``\val{payment\_type}''. In addition to the aforementioned constraint $Q_3: 0 \leq L_{inf}(C, C') \leq 0.01$, one may additionally require that this column to be at least 95\% complete (with less than 5\% of nulls), written as $Q_4: 0.95 \leq complete\_ratio(C) \leq 1$.  Furthermore, we expect to see no more than 6 distinct values (with ``\val{cash}'', ``\val{credit}'', etc.) in this column, so we have $Q_5: 0 \leq distinct\_cnt(C) \leq 6$.

Putting these together and let $\mathbf{S} = \{Q_3, Q_4, Q_5\}$, we can write a conjunctive program $P(\mathbf{S}) = \bigwedge_{Q_i \in \mathbf{S}} Q_i$ (or $Q_3 \land Q_4 \land Q_5$).
\end{example}

\section{\avhfull}
\label{sec:algo}

While DQ programs are flexible and powerful, they are difficult to write manually. We now describe our \avh to auto-program  DQ.

\subsection{Problem Statement}

For the scope of this work, we consider auto-generating conjunctive DQ programs for each column $C$ in data pipelines (or only for a subset of important columns selected by users), using column-level single-distribution or two-distribution DQ constraints (Section~\ref{sec:preliminary}). 

For a given column $C$, our goal is to program suitable DQ by selecting from a large space of metrics $\mathbf{M}$ in  Table~\ref{tab:numeric} and  Table~\ref{tab:category}. This space of $\mathbf{M}$ is clearly large and hard to program manually. Also note that while we list commonly-used metrics, the list is not meant to be exhaustive. In fact, \avh is designed to be \emph{extensible}, so that new metrics (e.g., statistical distances relevant to other use cases) can be added into $\mathbf{M}$ in a way that is transparent to users.


For a given $C$ and a set of possible $\mathbf{M}$, this induces a large space of possible DQ constraints on $C$. We denote this space of possible single-distribution and two-distribution DQ as $\mathbf{Q}$, defined as:
\begin{align} 
\begin{split}
\label{eqn:dq_space} 
\mathbf{Q} & =  \{Q(M, C, \theta_l, \theta_u) | M \in \mathbf{M}, \theta_l \in \mathbb{R}, \theta_u \in \mathbb{R}, \theta_l \leq \theta_u \} \\ 
& \cup  \{Q(M, C, C', \theta_l, \theta_u) | M \in \mathbf{M}, \theta_l \in \mathbb{R}, \theta_u \in \mathbb{R}, \theta_l \leq \theta_u \} 
\end{split}
\end{align}



We note that in production settings, it is crucial that auto-generated DQ programs are of high precision, with very few false-alarms (false-positive detection of DQ issues). This is because with thousands of recurring pipelines, even a low \emph{False-Positive Rate (FPR)} can translate into a large number of false-positives, which is undesirable as they usually  require human intervention. Because it is critical to ensure high precision, in \avh we explicitly aim for a very low level of FPR, which we denote by $\delta$, e.g., $\delta = 0.1\%$. 

Finally, because we are dealing with data from recurring data pipelines, we assume that the same data from $K$ past executions of this pipeline is available, which we denote as $H = \{C_1, C_2, \ldots, C_K\}$. These  $K$ previous batches of data are assumed to be free of DQ issues, which is reasonable because engineers usually manually check the first few pipeline runs after a pipeline is developed to ensure it runs properly. DQ issues tend to creep in over time due to data drift and schema drift~\cite{breck2019data, schelter2018automating}.

\textbf{\avhfull (\avh).}
Given these considerations, we now formally define our problem  as follows. 

\begin{definition} \avhfull.
Given a target column $C$ from a pipeline, and the same data from previous $K$ executions $H = \{C_1, C_2, \ldots, C_K\}$, a space of possible DQ constraints $\mathbf{Q}$, and a target false-positive-rate (FPR) $\delta$. Construct a conjunctive DQ program $P(\mathbf{S})$ with $\mathbf{S} \subseteq \mathbf{Q}$, such that the expected FPR of $P(\mathbf{S})$ is no greater than $\delta$, while $P(\mathbf{S})$ can catch as many DQ issues as possible. 
\end{definition}
We write \avh as the following optimization problem:
\begin{align}
\hspace{-1cm} \text{(\avh)} \qquad{} \max &~~ \text{R}(P(\mathbf{S})) \label{eqn:recall} \\  
\mbox{s.t.} ~~ & \text{FPR}(P(\mathbf{S})) \leq \delta \label{eqn:fpr}\\
  & P(\mathbf{S}) = \bigwedge_{Q_i \in  \mathbf{S},~~ \mathbf{S} \subseteq \mathbf{Q}} Q_i \label{eqn:conjunctive} 
\end{align}

Where $\text{R}(P(\mathbf{S}))$ denotes the expected recall of a DQ program $P(\mathbf{S})$ that we want to maximize, and $\text{FPR}(P(\mathbf{S}))$ denotes its expected FPR, which is required to be lower than a target threshold $\delta$. 


\subsection{Construct DQ constraints}
\label{sec:dq_constraint}

To solve \avh, in this section
we will first describe how to construct a large space of DQ constrains $\mathbf{Q}$ (and estimate their FPR), which are pre-requisites before we can use $\mathbf{Q}$ to generate conjunctive programs for \avh (in Section~\ref{sec:avh}).


Recall that to instantiate constraints like $Q_i(M, C, \theta_l, \theta_u)$ (from Definition~\ref{def:single}), we need to pick a metric $M \in \mathbf{M}$, apply $M$ on the given column $C$ to compute $M(C)$, and constrain $M(C)$ using suitable upper/lower-bounds thresholds $\theta_u$/$\theta_l$.

Here we leverage the fact that a history  $H = \{C_1, C_2, \ldots, C_K\}$ of the same column $C$ from past executions is available. If we apply $M$ on $H$, we obtain $M(H) = \{M(C_1), M(C_2), \ldots, M(C_K)\}$, which forms a statistical distribution\footnote{Later, we will discuss exceptions to the assumption (e.g., non-stationary time-series).}. When we apply the same metric $M$ on a newly arrived batch of data $C$, the resulting value $M(C)$ can then be seen as a data point drawn from the distribution $M(H)$. Let the estimated mean and variance of $M(H)$ be $\mu$ and $\sigma^2$, respectively. We can construct a DQ constraint $Q(M, C, \theta_l, \theta_u)$, with the following probabilistic FRP guarantees.

\begin{proposition}
\label{prop:chebyshev}
For any metric $M \in \mathbf{M}$, and $\beta \in \mathbb{R}^+$, we can construct a DQ constraint $Q(M, C, \theta_l, \theta_u)$, with $\theta_l = \mu - \beta,$  $\theta_u = \mu + \beta$. The expected FPR of the constructed $Q$ on data without DQ issues, denoted by $\mathrm{E}[FPR(Q)]$, satisfy the following inequality:
\begin{equation}
\label{eqn:chebyshev}
\mathrm{E}[FPR(Q)] \leq (\frac{\sigma}{\beta})^2
\end{equation}
\end{proposition}

\begin{proof}
We prove this proposition using Chebyshev's inequality~\cite{bulmer1979principles}. Chebyshev states that for a random variable $X$, $P(|X - \mu| \geq k \sigma) \leq \frac{1}{k^2}$, $\forall k \in \mathbb{R}^+$. For the random variable $M(C)$, let $k = \frac{\beta}{\sigma}$. Replacing $k$ with $\frac{\beta}{\sigma}$ above, we get $P(|M(C) - \mu| \geq \beta) \leq (\frac{\sigma}{\beta})^2$. Note that this implies  $P(|M(C) - \mu| \leq \beta) \geq 1 - (\frac{\sigma}{\beta})^2$, which can be rewritten as $P(-\beta \leq (M(C) - \mu) \leq \beta) \geq 1 - (\frac{\sigma}{\beta})^2$, or $P(\mu -\beta \leq M(C) \leq \mu + \beta) \geq 1 - (\frac{\sigma}{\beta})^2$. 

Observe that $\mu -\beta \leq M(C) \leq \mu + \beta$ is exactly our $Q(C, M, \theta_l, \theta_u)$, where $\theta_l = \mu - \beta$,  $\theta_u = \mu + \beta$. We thus get $P(Q \text{ holds on } C)$ $\geq$ $1 - (\frac{\sigma}{\beta})^2$, which is equivalent to saying that the expected FPR of $Q$ is no greater than $(\frac{\sigma}{\beta})^2$, or 
 $\mathrm{E}[FPR(Q)] \leq (\frac{\sigma}{\beta})^2$.
 \end{proof}

We use the following example to illustrate such a constructed DQ constraint and its estimated FPR.

\begin{example}
\label{ex:chebyshev}
Consider a metric $M = complete\_ratio$ from Table~\ref{tab:numeric} that computes the fraction of values in a column $C$ that are ``complete'' (not-null), and a history of data from past executions $H = \{C_1, C_2, \ldots, C_K\}$. Applying $M$ on $H$, we obtain the $complete\_ratio$ on historical data as $M(H) = \{0.92, 0.90, \ldots, 0.91\}$.

From the sample $M(H)$, we estimate its mean $\mu = 0.9$ and variance of $\sigma^2 = 0.0001$, respectively. Using Proposition~\ref{prop:chebyshev}, suppose we set $\beta = 0.05$, we get a $Q_6(C, complete\_ratio, 0.85, 0.95)$ (or equivalently  $0.85 \leq complete\_ratio(C) \leq 0.95$), whose expected FPR has the following inequality: 
 $\mathrm{E}[FPR(Q_6)] \leq (\frac{0.01}{0.05})^2 = 0.04$.
 
Note that using different $\beta$ allows us to instantiate different constraints with different levels of FPR. For example, setting $\beta = 0.1$ will induce a different $Q_7( complete\_ratio, C, 0.8, 1)$ (or $0.8 \leq complete\_ratio(C) \leq 1$), whose expected FPR is 
 $\mathrm{E}[FPR(Q_7)] \leq (\frac{0.01}{0.1})^2 = 0.01$. Note that this yields a lower FPR than that of $Q_6$ above, because $Q_7$ has a wider upper/lower-bound for $complete\_ratio$.
\end{example}

Using Proposition~\ref{prop:chebyshev} and different $\beta$ values, we can instantiate an array of DQ constraints using the same $M$ but different $[\theta_l, \theta_u]$ (thus different FPR guarantees). A DQ with a larger $\beta$ allows a larger range of $M(C)$ values, which is less sensitive/effective in catching DQ issues, but is also ``safer'' with lower expected FPR.

\textbf{Tighter bounds of FPR leveraging metric properties.}
The results in Proposition~\ref{prop:chebyshev} apply to any metric $M \in \mathbf{M}$, and the corresponding bounds on FRP are loose as a result. We derive two tighter FRP bounds for specific types of statistical metrics below, by exploiting unique characteristics of these metrics.

\begin{proposition}
\label{prop:cantelli}
For any metric $M \in \{EMD, JS\_div, KL\_div,$ $KS\_dist, Cohen\_d, L_1, L_{inf}, Cosine,$ $Chi\_squared\}$, and any $\beta \in \mathbb{R}^+$, we can construct a DQ constraint $Q(M, C, \theta_l, \theta_u)$, with $\theta_l = 0,$  $\theta_u = \mu + \beta$. 
 The expected FPR of the constructed $Q$ on data without DQ issues, denoted by $\mathrm{E}[FPR(Q)]$, satisfy the following inequality:
\begin{equation}
\label{eqn:cantelli}
\mathrm{E}[FPR(Q)] \leq \frac{\sigma^2}{\beta^2 + \sigma^2}
\end{equation}
\end{proposition}

This bound is derived using Cantelli's inequality~\cite{bulmer1979principles}, a proof of which can be found 
\iftoggle{fullversion}
{
    in Appendix~\ref{appendix:proof_cantelli}.
}
{
    in a full version of this paper~\cite{full}.
}

\begin{proposition}
\label{prop:clt}
For any metric $M \in \{count, mean,$ $str\_len,$ $char\_len, digit\_len, punc\_len, complete\_ratio\}$, and any $\beta \in \mathbb{R}^+$, we can construct a DQ constraint $Q(M, C, \theta_l, \theta_u)$, with $\theta_l = \mu - \beta,$  $\theta_u = \mu + \beta$. 
 The expected FPR of the constructed $Q$ on data without DQ issues, denoted by $\mathrm{E}[FPR(Q)]$, satisfy the following inequality:
\begin{equation}
\label{eqn:clt}
\mathrm{E}[FPR(Q)] \leq 1 - \frac{2}{\sqrt{\pi}}\int_{0}^{\frac{\beta}{\sqrt{2}\sigma}} e^{-t^2} dt
\end{equation}
\end{proposition}

This bound is derived using Central Limit Theorem~\cite{bulmer1979principles}.
\iftoggle{fullversion}
{
    We show a proof of this in Appendix~\ref{appendix:proof_clt}.
}
{
    (A proof of which can be found in~\cite{full}).
}

We omit examples for Proposition~\ref{prop:cantelli} and Proposition~\ref{prop:clt}, but DQ can be constructed similar to Example~\ref{ex:chebyshev} with tighter bounds.

We note that these tighter bounds allow us to construct DQ constraints with better FPR guarantees, which help to meet the constraint in Equation~\eqref{eqn:fpr} of \avh more effectively.
The pseudo-code of this step can be found in Appendix~\ref{apx:pseudo-code}.









\textbf{Time-series Differencing for Non-stationary Data.}
Our analysis so far assumes  $M(H)$ to be a well-behaved  distribution, generated from \emph{stationary processes}~\cite{park2018fundamentals}, defined as processes with probability distributions that are \emph{static} and do not change over time. While this is true for many real cases (e.g., Example~\ref{ex:chebyshev}), there are cases where  $M(H)$ follows \emph{non-stationary processes}~\cite{park2018fundamentals}, in which parameters of the underlying probability change over time.
\begin{example}
\label{ex:nonstationary}
Consider a recurring pipeline that processes one-day's worth of user traffic data visiting a website. Because overall, the user traffic will grow over time, the volume of data processed by the pipeline will increase slightly every day. So for the metric $M=row\_count$, we get a sequence of row-counts for the past $K$ days as $M(H) = \{ 100K, 103K, 105K, 106K, \ldots, 151K, 152K\}$. Note that $M(H)$ is \emph{non-stationary} here, because the parameters of the underlying distribution (e.g., the mean of $M(C)$) change over time.
\end{example}

Modeling non-stationary $M(H)$ like above as stationary using a static distribution is clearly sub-optimal, which may lead to false-positives and false-negatives in DQ applications.

To account for non-stationary $M(H)$, we first determine whether a $M(H)$ is already stationary, using the Augmented Dickey–Fuller (ADF) test from the time-series literature~\cite{cheung1995lag}. If we reject the null-hypothesis in ADF that $M(H)$ is already stationary (e.g., Example~\ref{ex:chebyshev}), we proceed to construct DQ constraints as before. For cases where $M(H)$ is not stationary (e.g., Example~\ref{ex:nonstationary}), we repeatedly apply a technique known as  time-series differencing~\cite{granger1980introduction} on  $M(H)$ until it reaches stationarity. 
\iftoggle{fullversion}
{
    we illustrate this using a small example below, and defer details of the time-series differencing step to Appendix~\ref{appendix:differencing}.
}
{
    While we defer details of the time-series differencing step to a full version of the paper~\cite{full} due to space limit, we illustrate this using a small example below.
}

\begin{example}
\label{ex:differencing}
Continue with Example~\ref{ex:nonstationary}, where $M(H) = \{ 100K,$ $103K, 105K, 106K, \ldots, 151K, 153K\}$, and the metric $M=row\_count$. The Augmented Dickey–Fuller (ADF) test will fail to reject the null hypothesis that $M(H)$ is non-stationary. Applying a first-order time-differencing step~(\cite{park2018fundamentals}) with $t=1$ will produce: $M'_{t=1}(H) = \{ M(C_{2}) - M(C_1), M(C_{3}) - M(C_2), \ldots$ $M(C_{K}) - M(C_{K-1}),\}$ $=$ $\{ 3K, 2K,$ $1K,$ $\ldots, 2K\}$. This resulting $M'_{t=1}(H)$ passes the ADF test and is
then used as a static distribution to generate $\mathbf{Q}$.
\end{example}

\iftoggle{fullversion}
{
    We note that
}
{
    We defer details of this step to~\cite{full}, but we note that
}
the differencing step also allows us to handle cyclic time-series $M(H)$  (e.g., weekly or hourly periodic patterns), by transforming $M(H)$ using first-order differencing with lags~\cite{granger1980introduction}, which can then be handled like stationary processes as before.

\vspace{-2mm}
\subsection{Construct DQ Programs in \avh} \hfill\\
\label{sec:avh}

After we construct constraints $\mathbf{Q}$ and estimate their FPR bounds, 
we are ready solve \avh. 
Recall that in \avh, in addition to satisfying the hard constraint on FPR (Equation~\eqref{eqn:fpr}), our objective (Equation~\eqref{eqn:recall}) is to maximize the expected ``recall'' of the constructed DQ program  (the number of possible DQ issues to catch).
In order to fully instantiate \avh, we still need to estimate the expected recall benefit of each DQ constraint $Q_i \in \mathbf{Q}$, which can guide us to select the most ``beneficial'' DQ program.

\textbf{Estimate DQ recall using synthetic ``training''.} 
Clearly, we cannot foresee the exact DQ issues that may arise in the future in a particular pipeline, to precisely quantify the benefit of each $Q_i \in \mathbf{Q}$. However, 
there is a large literature that documents common types of DQ issues in pipelines (e.g.,~\cite{breck2019data, schelter2018automating, schelter2018deequ, schelter2019differential, swami2020data, polyzotis2017data, song2021auto}), which include things like schema change, unit change, increased nulls, as discussed earlier. Our observation is that although it is hard to quantify the benefit of $Q_i$ in a specific DQ incident, in the long run if future DQ issues are drawn from the set of common DQ problems, then we can still estimate the expected recall of a specific $Q_i$.

With that goal in mind, we carefully reviewed the DQ literature and cataloged a list of 10 common types of DQ issues in pipelines (schema change, unit change, increased nulls, etc.). We then vary parameters in each type of DQ to systematically capture different magnitudes of DQ deviations (e.g., different fractions of values are overwritten with nulls for ``increased nulls'', different magnitudes of changes for ``unit changes'', etc.), to construct a total of 60 procedures that can systematically inject DQ issues in a given column $C$ by varying $C$. We denote this set of synthetically generated DQ issues on $C$ as $\mathbf{D}(C)$. 
\iftoggle{fullversion}
{
    We give a full list of these common types of DQ issues and their parameters configurations in Appendix~\ref{apx:datagen}.
}
{
    In the interest of space, we defer a detailed description of these common DQ issues, and the corresponding procedures to generate $\mathbf{D}(C)$, to a full version of the paper~\cite{full}.
}

Intuitively, this $\mathbf{D}(C)$ models a wide variety of data deviations that may happen in $C$ due to DQ issues, which guides us to select salient $Q_i \in \mathbf{Q}$ that are unique ``statistical invariants'' specific to a pipeline, to best differentiate between the ``normal'' $H$, and the ``bad cases'' in $\mathbf{D}(C)$, for this pipeline.  This synthetic $\mathbf{D}(C)$  in effect becomes ``training data'' in ML, by assisting us to estimate the recall benefit of $Q_i$ in \avh. We give an example below to illustrate this.

\begin{example}
\label{ex:synthetic}
We revisit the example from the Introduction, where a recurring pipeline produces exactly 50 output rows, with one row for each of the 50 US states,  over all past $K$ executions in the history. In such a pipeline, for the ``\val{state}'' column in the output, one distinguishing feature is that the column has exactly 50 distinct values, or $Q_8: dist\_val\_cnt(\text{\val{state}}) = 50$ (which intuitively, is a ``statistical invariant'' only unique to this pipeline). 

When we synthetically inject DQ issues into $C$ (the ``\val{state}'' column) to produce $\mathbf{D}(C)$, we get variants of $C$, such as $C$ with an increased number of nulls, $C$ with values taken from a neighboring column (due to schema-change), etc. This constraint $Q_8: dist\_val\_cnt(\text{\val{state}})  = 50$ will catch most of such variations in $\mathbf{D}(C)$, thus producing a high expected ``recall'' and making $Q_8$ a desirable constraint to use. Intuitively, $Q_8$ is a good constraint for $C$ in this particular pipeline, because $dist\_val\_cnt = 50$ is a unique ``statistical variant'' specific to this column and pipeline, which has more discriminating power than other more generic constraints. 
\end{example}

Formally, we define the expected recall of $Q_i$ or $R(Q_i)$, as the set of issues it can detect in $\mathbf{D}(C)$, written as:
\begin{equation}
    R(Q_i) = \{C' | C' \in \mathbf{D}(C),  C' \text{~fails on~} Q_i\}
\end{equation} 

\textbf{Optimizing \avh with guarantees.}
Given a DQ program  with a conjunction of constraints $P(\mathbf{S}) = \bigwedge_{Q_i \in  \mathbf{S}} Q_i$ for some $\mathbf{S} \subseteq \mathbf{Q}$, naturally the recall of two constraints $Q_i, Q_j \in \mathbf{S}$ will overlap (with $ R(Q_i) \cup  R(Q_j) \neq \emptyset$). This leads to diminishing recall for similar DQ constraints in the same program, and requires us to leverage ``complementary'' constraints when generating DQ programs.

Given a conjunctive program $P(\mathbf{S})$ with $\mathbf{S} \subseteq \mathbf{Q}$, we  model the collective recall of $\mathbf{S}$, as the union of individual $R(Q_i)$, or $\bigcup_{Q_i \in \mathbf{S}}{ \text{R}(Q_i) }$. This becomes a concrete instantiation of the objective function in Equation~\eqref{eqn:recall} of the \avh problem.

Furthermore, recall that we can upper-bound the FPR of each $Q_i \in \mathbf{S}$, using Proposition~\ref{prop:chebyshev}-\ref{prop:clt}. Given a program $P(\mathbf{S})$, assume a worst-case where the FRP of each $Q_i \in \mathbf{S}$ is disjoint, we can then upper-bound the FPR of $P(\mathbf{S})$ (Equation~\eqref{eqn:fpr}), as the sum of the FRP bounds of each $Q_i$, or
$\sum_{Q_i \in \mathbf{S}} \text{FPR}(Q_i)$. 

Together, we rewrite the abstract \avh in Equation~\eqref{eqn:recall}-\eqref{eqn:conjunctive} as:
\begin{align}
\hspace{-1cm} \text{(\avh)} \qquad{} \max &~~ \bigg| \bigcup_{Q_i \in \mathbf{S}}{ \text{R}(Q_i) } \bigg| 
\label{eqn:recall-2} \\  
\mbox{s.t.} ~~ & \sum_{Q_i \in \mathbf{S}} \text{FPR}(Q_i) \leq \delta \label{eqn:fpr-2}\\
  & \mathbf{S} \subseteq \mathbf{Q} \label{eqn:conjunctive-2} 
\end{align}

Intuitively, we want to weigh the ``cost'' of selecting a constraint $Q_i$, which is its estimated  $FRP(Q_i)$, against its ``benefit'', which is its expected recall $R(Q_i)$. Furthermore, we need to account for the fact that constraints with overlapping recall benefits yield diminishing returns that is analogous to submodularity.
\iftoggle{fullversion}
{
    We prove that \avh is in general intractable and hard to approximate in Appendix~\ref{appendix:hardness}.
}
{
    (we prove  \avh is NP-hard in~\cite{full}).
}

Given that it is unlikely that we can solve \avh optimally in polynomial time, we propose an efficient algorithm that gives a constant-factor approximation of the best possible solution in terms of the objective value in Equation~\eqref{eqn:recall-2}, while still guaranteed to satisfy the FPR requirement in Equation~\eqref{eqn:fpr-2} in expectation. The pseudo-code of the procedure is shown in Algorithm~\ref{algo:avh}.

\begin{algorithm}[t]
\SetKw{kwReturn}{return}
 \Input{Metrics $\mathbf{M}$, a target-FPR $\delta$, column $C$, and its history $H = \{C_1, C_2, \ldots, C_K\}$}
 \Output{Conjunctive DQ Program $\mathbf{P(S)}$}
 
 $\mathbf{Q} \leftarrow$ Construct-Constraints$(\mathbf{M}, H)$ 


$S \leftarrow \emptyset$, $FPR \leftarrow 0$

  \While{$FPR \leq \delta$}
    {
       $Q_s=\arg max_{Q_i \in \mathbf{Q}} (\frac{| R(Q_i) \setminus \bigcup_{Q_j \in \mathbf{S}}R(Q_j) |}{FRP(Q_i)})$
       
     \If{$FPR(Q_s) + FPR \leq \delta$}
         {
             $S \gets S\cup Q_s$\\
             $FPR \gets FPR + FPR(Q_s)$
         }
  
     $\mathbf{Q} \gets \mathbf{Q} \backslash Q_s$ 
    }

$Q_m = \arg max_{Q_m \in \mathbf{Q}} (|R(Q_m)|)$

\If{$| \bigcup_{Q_i \in \mathbf{S}}{ \text{R}(Q_i) } |  < |R(Q_m)|$}
    {
    $S \leftarrow \{Q_m\}$
    }
\kwReturn $P(S)$
\caption{Auto-Validate by-History (\avh)}
\label{algo:avh}
\end{algorithm}

Algorithm~\ref{algo:avh} takes as input a set of metrics $\mathbf{M}$, an FPR target $\delta$, as well as a column $C$ together with its history $H = \{C_1, C_2, \ldots, C_K\}$. We start by constructing a large space of possible DQ constraints $\mathbf{Q}$  (Line 1), 
using the given $\mathbf{M}$ and $H$ (Section~\ref{sec:dq_constraint}).

Using this $\mathbf{Q}$, we then iterate to find a solution $\mathbf{S} \subseteq \mathbf{Q}$, which is first initialized to empty. In each iteration, we select the best possible $Q_s$ from remaining constraints in $\mathbf{Q}$ that have not yet been selected (Line 4), based on a cost/benefit calculation, where the ``benefit'' of adding a constraint $Q_i$ is its increment recall gain on top of the current solution set $\mathbf{S}$, written as  $| R(Q_i) \setminus \bigcup_{Q_j \in \mathbf{S}}R(Q_j) |$, divided by its additional ``cost'' of adding $Q_i$, which is the increased FPR when adding $FPR(Q_i)$. The selected $Q_s$ is then simply the constraint that maximizes this benefit-to-cost ratio, as shown in Line 4. We add this $Q_s$ to the current solution $\mathbf{S}$, update the current total FPR as well $\mathbf{Q}$, and iterate until we exhaust $\mathbf{Q}$. 

In the final step (Line 9), we compare the best possible singleton $Q_m \in \mathbf{Q}$ that maximizes recall without violating the FPR requirement, with the current $\mathbf{S}$ from above. We pick the best between $\{Q_m\}$ and $\mathbf{S}$ based on their recall as our final solution to \avh. 

We show that Algorithm~\ref{algo:avh} has the following properties
\iftoggle{fullversion}
{
    (a proof of which can be found in Appendix~\ref{appendix:approx}).
}
{
    a proof of which can be found in~\cite{full} in the interest of space. 
}

\begin{proposition}
\label{prop:approx}
Algorithm~\ref{algo:avh} is a $(\frac{1}{2}-\frac{1}{2e})$-approximation algorithm for the \avh problem in Equation~\eqref{eqn:recall-2}, meaning that the  objective value produced by  Algorithm~\ref{algo:avh} is at least $(\frac{1}{2}-\frac{1}{2e}) \text{OPT}$, when $\text{OPT}$ is the objective value of the optimal solution to \avh. Furthermore, 
Algorithm~\ref{algo:avh} produces a feasible solution in expectation, meaning that the expected FPR of its solution is guaranteed to satisfy Equation~\eqref{eqn:fpr-2}.
\end{proposition}


\vspace{-3mm}
\section{Experiments}
We evaluate the effectiveness and efficiency of \avh, using real production pipelines. {Our code will be shared at~\cite{code} after an internal review.}

\subsection{Evaluation Benchmarks}
\label{sec:benchmark}

\textbf{Benchmarks.} We perform rigorous evaluations, using real and synthetic benchmarks derived from production pipelines.

\textbf{- \textsc{Real}}. We construct a \textsc{Real} benchmark using production pipelines from Microsoft's 
internal big-data platform 
~\cite{zhou2012scope}.
We perform a longitudinal study of the pipelines, by sampling 1000 numeric columns and 1000 categorical columns from these recurring pipelines, and trace them over 60 consecutive executions  (which may recur daily or hourly). For each column $C$, this generates a sequence of history $\{C_1, C_2, \ldots, C_{60}\}$, for a total of 2000 sequences.  

We evaluate the precision/recall of each algorithm $\mathcal{A}$ (\avh or otherwise) on the 2000 sequences, by constructing sliding windows of sub-sequences for back-in-time tests of $\mathcal{A}$'s precision/recall (following similar practices in other time-series domains~\cite{campbell2005review, chatfield2000time}):

\underline{Precision}. Given a sequence of past runs $H = \{C_1, C_2, \ldots, C_K\}$, if an algorithm $\mathcal{A}$ looks at $H$ together with the real $C_{K+1}$ that arrives next, and predicts $C_{K+1}$ to have data-quality issues, then it is likely a false-positive detection, because the vast majority of production pipeline runs are free of DQ issues (if there were anomalous runs, they would have been caught and fixed by engineers, given the importance of the production data). To validate that it is indeed the case in our test data, we manually inspected a sample of our production pipeline data and did not identify any DQ issues. (Details of the process can be found  
    \iftoggle{fullversion}
    {
       in Appendix~\ref{sec:error}.)
    }
    {
        in a full version of the paper~\cite{full}.)
    }

For each full sequence $S = \{C_1, C_2, \ldots, C_{60}\}$, we construct a total of 30 historical sliding windows (each with a length of 30), as $H_{30} = \{C_1, C_2, \ldots, C_{30}\}$, $H_{31} = \{C_2, C_3, \ldots, C_{31}\}$, etc. Then at time-step $K$ (e.g., 30), and given the history $H_K = \{C_{K-29},$ $C_{K-28},$ $\ldots, C_{K}\}$, we ask each algorithm $\mathcal{A}$ to look at $H_K$ and predict whether the next batch of real data $C_{K+1}$ has a DQ issue or not, for a total of ($2000 \times 30) = $ 60K precision tests.

\underline{Recall}. For recall, because there are few documented DQ incidents that we can use to test algorithms at scale, we systematically construct recall tests as follows. Given a sliding window of prefix $H = \{C_1,$ $C_2,$ $\ldots, C_{30}\}$, we swap out the next batch of real data $C_{31}$, and replace it with a column $C'_{31}$ that looks ``similar'' to $C_{31}$ (e.g., with a similar set of values). 

Specifically, we use $C_{31}$ as the ``seed query'', to retrieve top-20 columns most similar to $C_{31}$ based on content similarity (Cosine), from the underlying data lake that hosts all production pipelines. 
Because $C'_{31}$ will likely have subtle differences from the real $C_{31}$ (e.g., value-distributions, row-counts, etc.), algorithm $\mathcal{A}$ should ideally detect as many $C'_{31}$ as DQ issues as possible (good recall), without triggering false alarms on the real $C_{31}$ (good precision). Because we retrieve top-20 similar columns, this generates a total of $2000 \times 20 = $ 40K recall tests.\footnote{It should be noted that some of the $C'_{31}$ columns we retrieve may be so similar to  $C_{31}$ that they become indistinguishable, making it impossible for any $\mathcal{A}$ to detect such  $C'_{31}$ as DQ issues. This lowers the best-possible recall, but is fair to all algorithms.}

\textbf{-\textsc{Synthetic}}. In addition, we create a \textsc{Synthetic} benchmark, where the precision tests are identical to \textsc{Real}. For recall tests, instead of using real columns that are similar to $C_{31}$, we synthetically inject 10 common DQ issues reported in the literature into $C_{31}$ (described in Appendix~\ref{apx:datagen}). This allows us to systematically test against a range of DQ issues with different levels of deviations.

\textbf{Evaluation metrics.} For each algorithm $\mathcal{A}$, we report standard precision/recall results on the 60K precision tests and 40K recall tests described above. We use standard precision and recall, defined as $precision = \frac{TP}{TP + FP}$,  $recall = \frac{TP}{TP + FN}$, where TP, FP, and FN are True-Positive, False-Positive, and False-Negative, respectively.

\subsection{Methods Compared}
\label{subsec:methods}
We compare with an extensive set of over 20 methods, including strong commercial solutions, as well as state-of-the-art  algorithms from the literature of anomaly detection and data cleaning. We categorize these methods into groups, which we describe  below.

\textbf{Commercial solutions}. We compare with the following commercial solutions that aim to automatically validate data pipelines.
\underline{Google TFDV}. We compare with Google's 
Tensorflow Data Validation (TFDV)~\cite{caveness2020tensorflow}. We install the latest version from Python pip,
and use recommended settings in~\cite{tfdv-settings}. 

\noindent \underline{Amazon Deequ}. We compare with Amazon's Deequ~\cite{schelter2018deequ, schelter2019unit}, 
using configurations suggested in their documentations~\cite{deequ-settings}. 

\noindent \underline{Azure Anomaly Detector}. Azure Anomaly 
Detector~\cite{Azure-Anomaly} is a cloud-based anomaly 
detection service for time-series data, utilizing state-of-the-art algorithms in the literature~\cite{ren2019time}.

\noindent \underline{Azure ML Drift Detection}. Azure ML has the ability to detect data drift over time~\cite{aml_drift}. We use data from the past $K$ executions as the ``baseline'' and data from a new execution as the ``target''.

\textbf{Time-Series-based anomaly detection}. There is a large body of literature on detecting outliers from time-series data. We use a recent benchmark study~\cite{timeseries-anomaly-survey} to identify the following four best-performing methods, and use the same implementations provided in~\cite{time-series-code} on our statistical data for comparison purposes.

\noindent  \underline{LSTM-AD}~\cite{malhotra2016lstm} employs LSTM networks to learn and reconstruct time series. It uses the reconstruction error to detect anomalies.

\noindent \underline{Telemanom}~\cite{hundman2018detecting} also uses LSTM networks to reconstruct time-series telemetry, identifying anomalies by comparing expected and actual values and applying unsupervised thresholds. 

\noindent \underline{Health-ESN}~\cite{chen2020imbalanced} uses the classical Echo State Network (ESN) and is trained on normal data. Anomalies are detected when the error between the input and predicted output exceeds a certain threshold, which is determined through an information theoretic analysis. 

\noindent \underline{COF}~\cite{tang2002enhancing} is a local density-based method that identifies time-series outliers, by detecting deviations from spherical density patterns.

\textbf{Classical anomaly detection.} We also compare with the following anomaly detection methods developed in tabular settings.

\noindent \underline{One-class SVM}~\cite{scholkopf1999support} is a popular ML method for anomaly detection, where only one class of training data is available. 
We train one-class SVM using historical data, and use it to make predictions. 

 \noindent  \underline{Isolation Forest}~\cite{liu2008isolation} is also a popular method for anomaly detection based on decision trees. We again train Isolation Forest using historical data, and then predict on newly-arrived data. 
 
\noindent \underline{Local Outlier Factor (LOF)}~\cite{breunig2000lof} is another one-class method for anomaly detection based on data density. We configure LOF in a way similar to other one-class methods above. 

\noindent \underline{K-MeansAD}~\cite{lima2010anomaly} is also a classical anomaly detection method, which is based on the unsupervised K-Means clustering. 

\noindent  \underline{ECOD}~\cite{li2022ecod} identifies outliers by estimating the distribution of the input data and calculating the tail probability for each data point.

\noindent  \underline{Average KNN (Avg-KNN)} is another outlier detection method, and was used to automate data validation in a pipeline setting~\cite{redyuk2021automating} that is similar to the scenario considered in our work.


\textbf{Statistical tests.} We compare with the following classical statistical tests used to detect outliers in distributions.

\noindent \underline{Kolmogorov–Smirnov (KS)} is a classical statistical hypothesis test for homogeneity between two numeric distributions, and is used in prior work to detect data drift~\cite{polyzotis2017data}. We vary its p-value thresholds to generate PR curves.  

\noindent \underline{Chi-squared} is a classical hypothesis test for homogeneity between two categorical distributions, and also used in prior work~\cite{polyzotis2017data}. We vary its p-value thresholds like above. 

\noindent \underline{Median Absolute Deviation (MAD)} is a measure of statistical dispersion from robust statistics~\cite{huber2011robust}, and has been used to detect quantitative outliers~\cite{hellerstein2008quantitative}. We use  MAD-deviation (Hampel X84, similar to z-scores)  to produce predictions~\cite{hellerstein2008quantitative}.

\textbf{Database constraints}. There is a large literature on using database constraints for data cleaning. We compare with these methods:

\noindent \underline{Functional Dependency (FD)}. FD is widely-used to detect data errors in tables~\cite{mahdavi2019raha, heidari2019holodetect, beskales2014sampling, papenbrock2016hybrid}, by exploiting correlations between columns (e.g., salary $\rightarrow$ tax-rate). Since not all columns can be ``covered'' by FD, to estimate its best possible recall, we detect all possible FDs from our 2000 test tables, and mark a test column $C$ to be ``covered'' if there exists a detected FD that has $C$ in its RHS. We report this as FD-UB (Functional Dependency Upper-bound). 

\noindent \underline{Order Dependency (OD)~\cite{szlichta2016effective, langer2016efficient, szlichta2013expressiveness}}. We discover OD using the same statistical information by ordering tables with statics in time.

\noindent \underline{Sequential Dependency (SD)}~\cite{golab2009sequential, caruccio2015relaxed} generalizes OD, and we discover SD using the same statistical information over time.

\noindent \underline{Denial Constraints (DC)}. We use the approach in~\cite{chu2013discovering} to discover DC that generalizes FD and OD, and use them for validating data. 



\textbf{\avhfull{} (\avh)}. This is our proposed \avh method as described in Section~\ref{sec:algo}.

\begin{figure*}[t!]
\vspace{-16mm}
    \centering
    \subfigure[Test on numerical data]{
        \label{fig:num-test on real}
        \includegraphics[height=7.5cm]{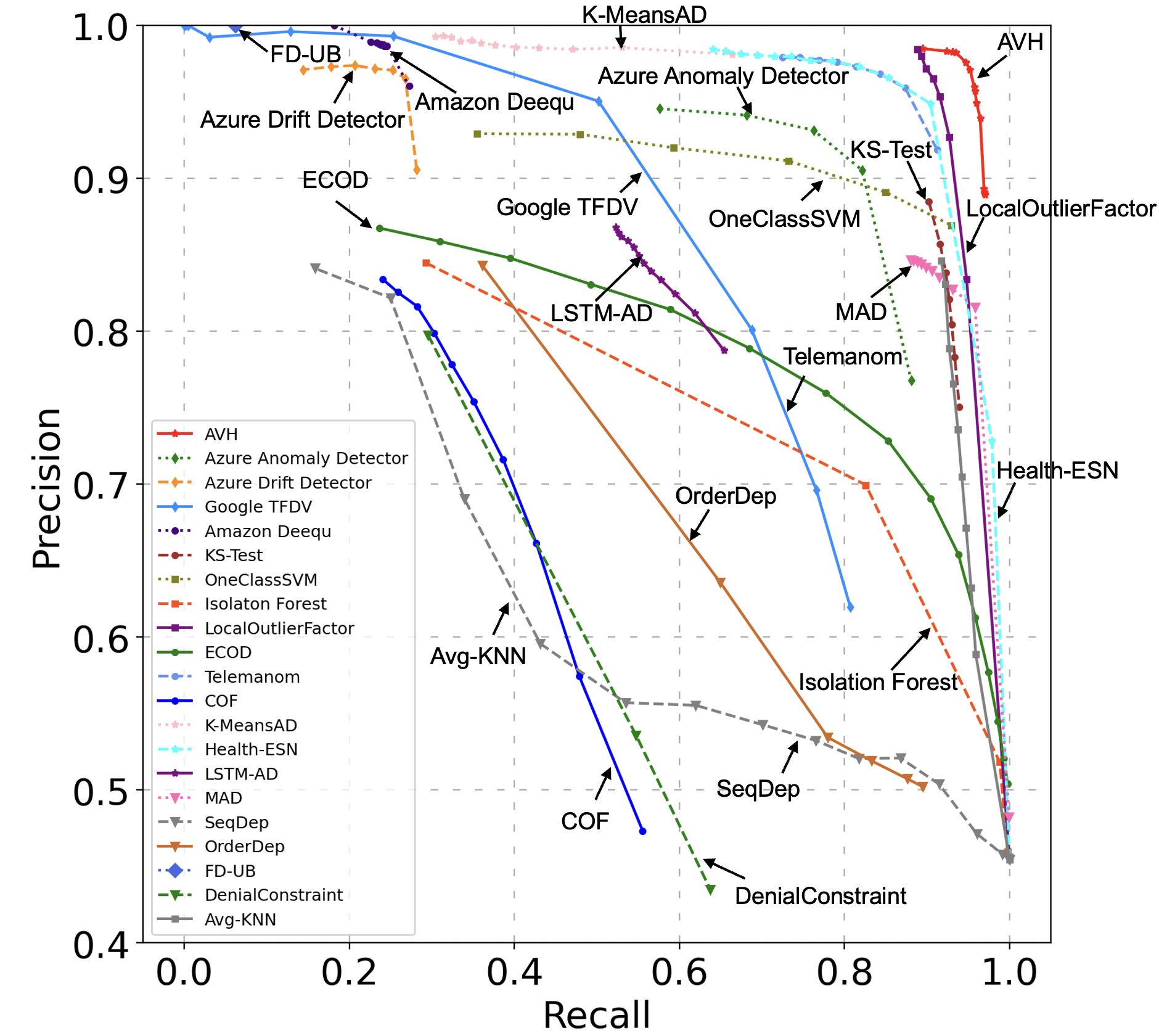}
    } 
    \subfigure[Test on categorical data]{
        \label{fig:num-test on syn}
        \includegraphics[height=7.5cm]{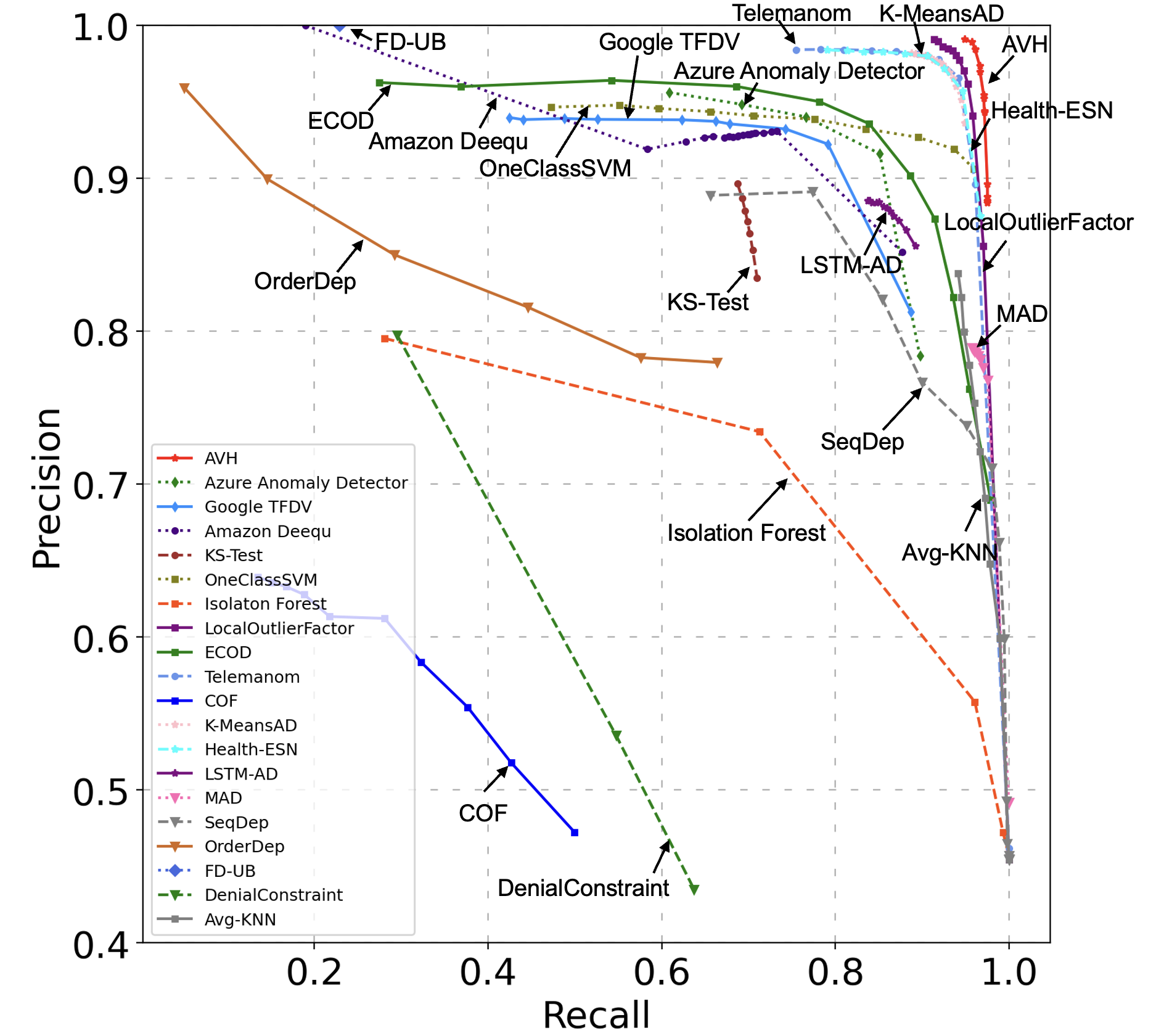}
    } \vspace{-5mm}
    \caption{{Precision/Recall results on the \textsc{Real} benchmark (2000 real pipelines).}}
    \label{fig:real}
    \vspace{-2mm}
\end{figure*}

\begin{figure*}[t!]
\vspace{-3mm}
    \centering
    \hspace{-3mm}
    \subfigure[Test on numerical data]{
        \label{fig:cat-test on real}
        \includegraphics[height=7.5cm]{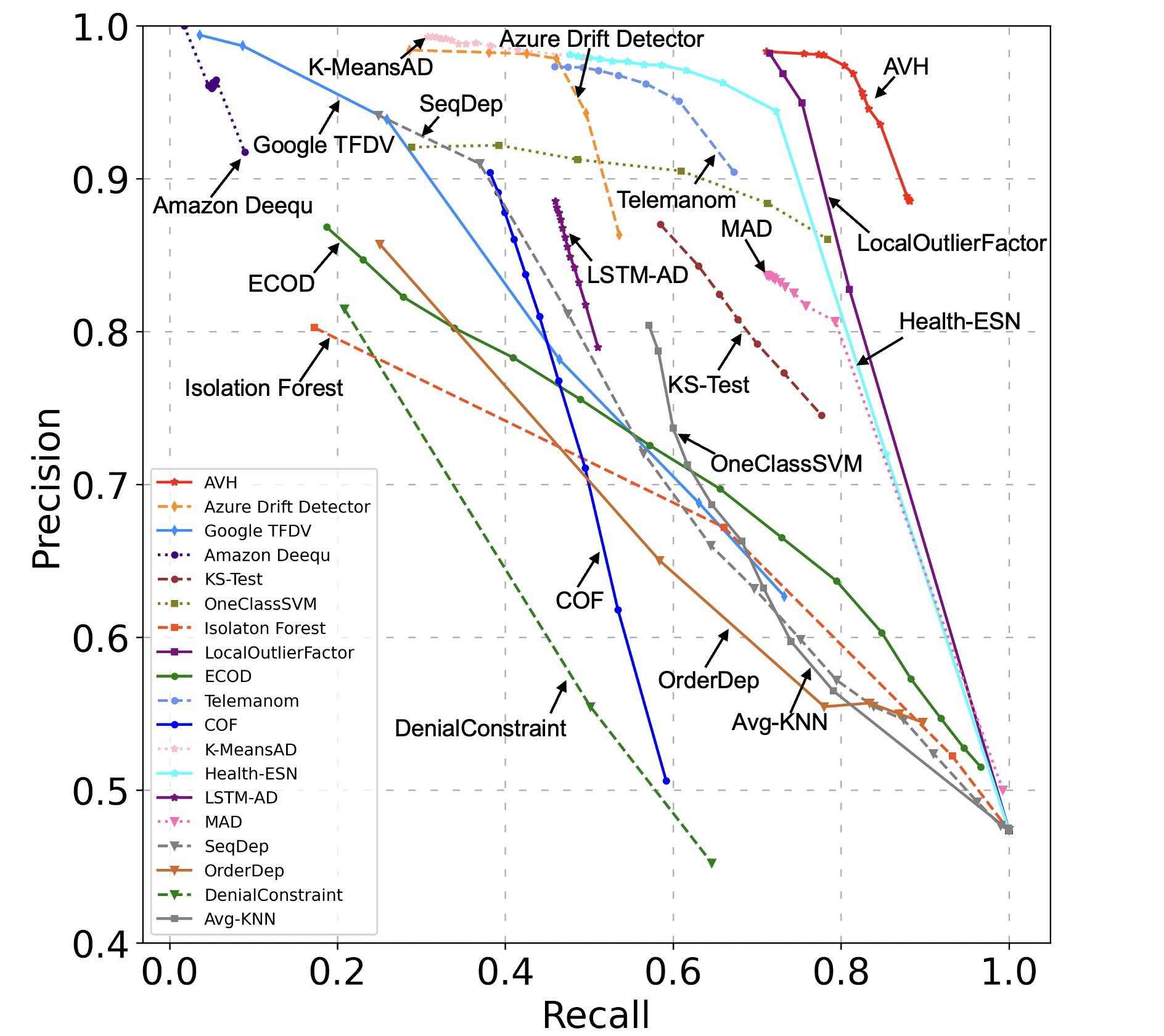}} 
    \hspace{1mm}
    \subfigure[Test on categorical data]{
        \label{fig:cat-test on syn}
        \includegraphics[height=7.5cm]{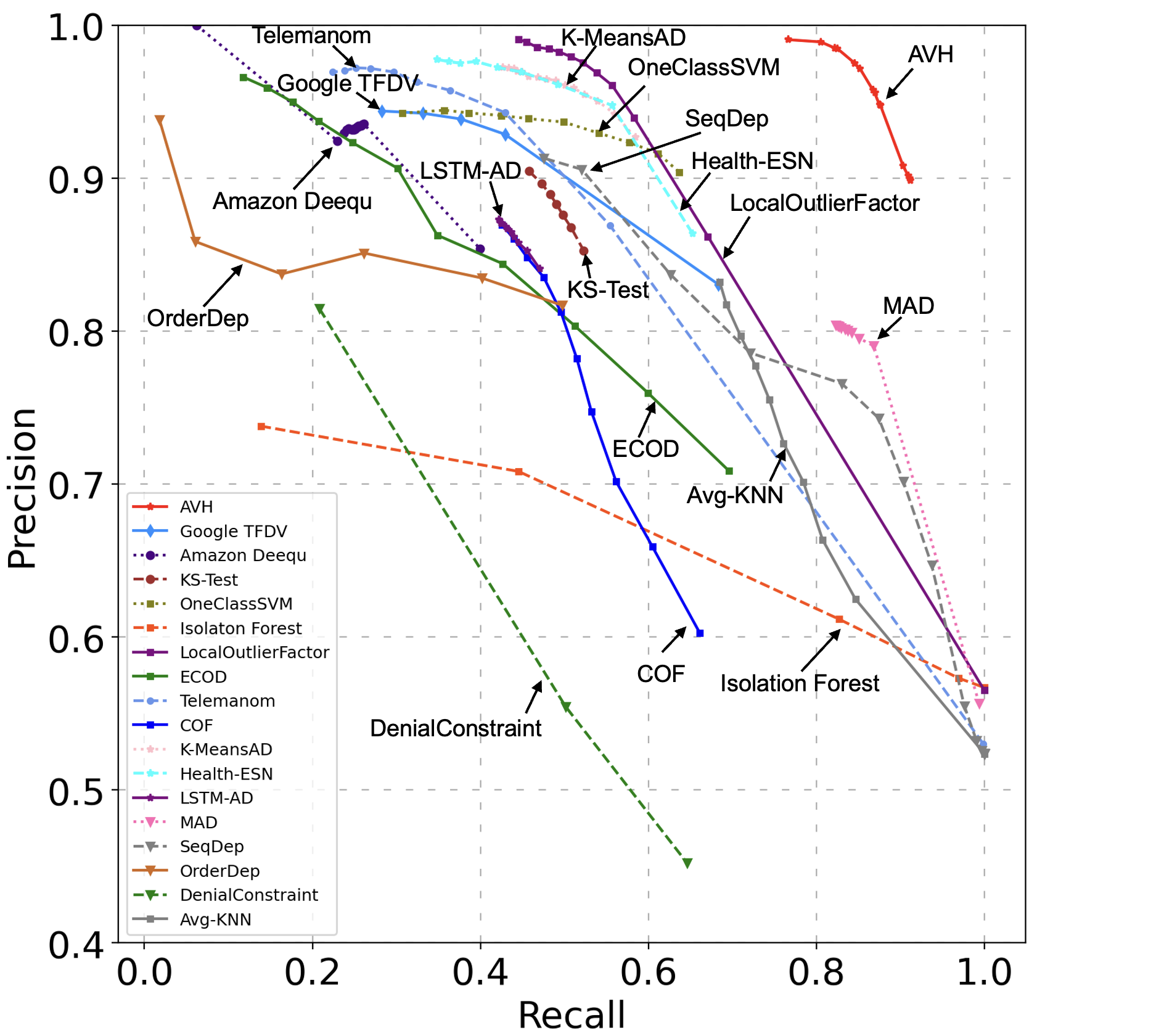}}
    \vspace{-5mm}
    \caption{Precision/Recall results on the \texttt{Synthetic} benchmark.}
    \label{fig:synthetic-all}
\end{figure*}

\vspace{-2mm}
\subsection{Experiment Results}
\textbf{Overall quality comparisons.}
Figure \ref{fig:real} and \ref{fig:synthetic-all} show the average precision/recall of different methods on the \textsc{Real} and \textsc{Synthetic} benchmark, respectively. \avh is at the top-right corner with high precision/recall, outperforming other methods across all cases.

Anomaly detection methods, especially LOF and Health-ESP, are the best performing baselines. However, these methods use each statistical attribute just as a regular dimension in a data record, \textit{while  our \avh exploits different statistical properties of the underlying metrics} (Chebyshev, Chantelli, CLT, etc., in Proposition~\ref{prop:chebyshev}-\ref{prop:clt}), which gives \avh an unique advantage over even the state-of-the-art anomaly detection methods, underscoring the importance of our approach in validating data from recurring data pipelines.

Commercial data-validation solutions like Amazon Deequ and Google TFDV have high precision but low recall, because they use  predefined and static configurations (e.g., JS-Divergence and L-infinity are the default for TFDV), which lack the ability adapt to different pipelines, and thus the low recall. 
Similarly, statistical tests (KS/Chi-squared/MAD) use fixed predictors that also cannot adapt to different pipelines, and show sub-optimal performance.

Constraint-based data cleaning methods from the database literature (e.g., FD, DC, OD, SD, etc.) are not competitive in our tests, because these methods are  designed to handle single table snapshot, typically using manually designed constraints.

Figure~\ref{fig:syn-by-types} shows breakdown of \avh results in Figure~\ref{fig:synthetic-all} by different types of DQ issues  in the \textsc{Synthetic} benchmark. We can see that \avh is effective against most types of DQ issues (schema-change, distribution-change, data-volume-change, etc.). On numerical data, we see that it is the most difficult to detect ``character-level perturbation'' (randomly perturbing one digit character for another digit with small probabilities) and ``character deletion'' (randomly removing one digit character with small probabilities), which is not unexpected since such small changes may not always change the underlying numerical distributions. 
On categorical data, ``character-level perturbation'' is also the most difficult to detect, but \avh is effective against ``character deletion'' and ``character insertion''. 

\textbf{Sensitivity and ablation studies.} We perform extensive experiments to study the sensitivity of \avh (to the length of history, different types of data-errors, target FPR $\tau$, etc.). We also perform an ablation study to understand the importance of \avh components. In the interest of space, we present these additional experimental results in Appendix~\ref{apx:exp-sensitivity} and Appendix~\ref{apx:exp-ablation}, respectively.

\begin{figure}[t]            
\vspace{-4mm}
    \centering\includegraphics[height=3.5cm]{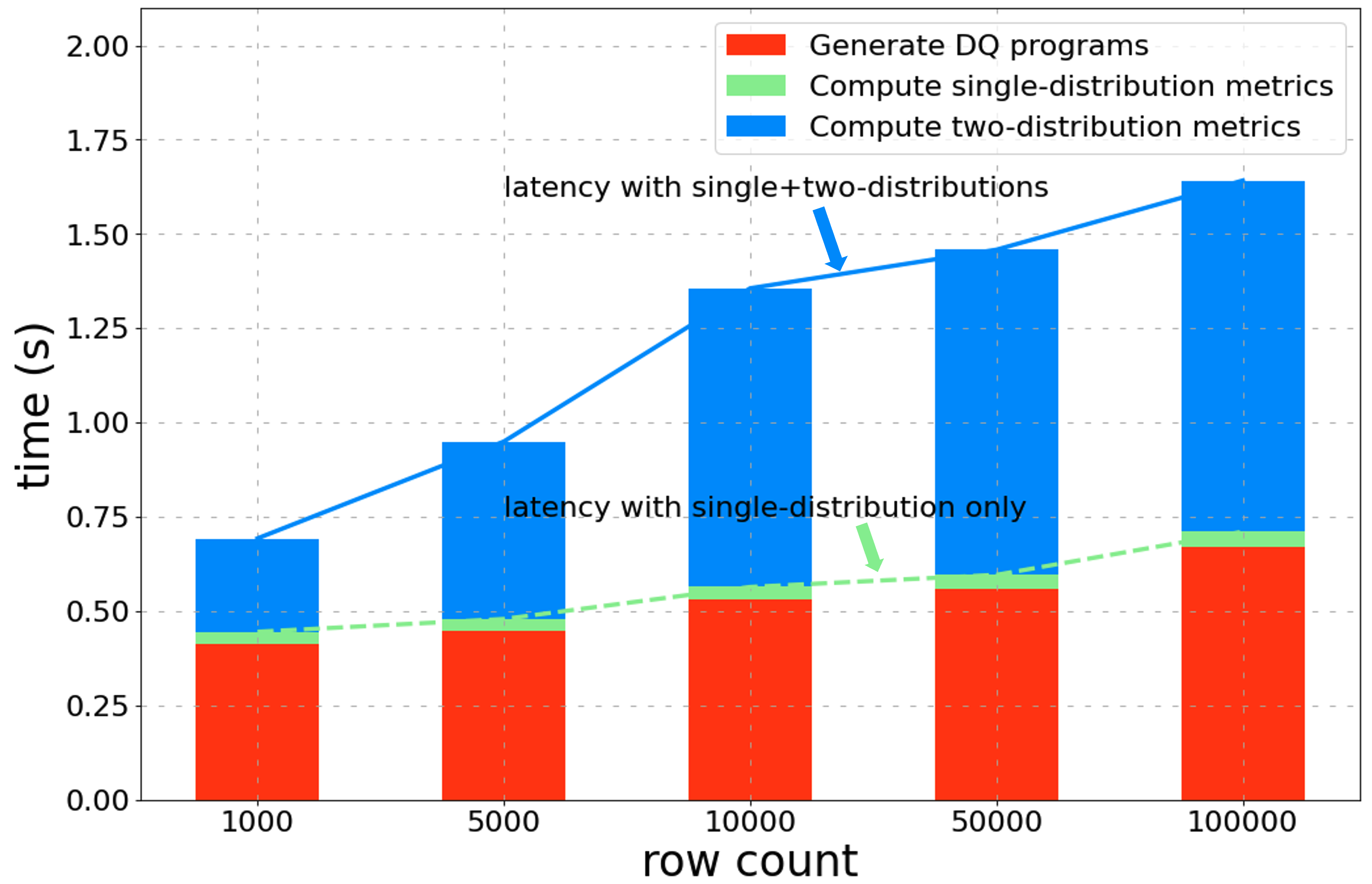}
    \vspace{-4mm}
    \caption{{Efficiency analysis of \avh}}
    \label{fig:efficiency}
    \vspace{-5mm}
\end{figure}

\textbf{Efficiency.}
Figure~\ref{fig:efficiency} shows the end-to-end latency of \avh to process a new batch of data. We vary the number of rows in a column $C$ (x-axis), and report latency averaged over 100 runs.
Recall that \avh can be used offline, since DQ constraints can be auto-installed on recurring pipelines (without involving humans). Nevertheless, we want to make sure that the cost of \avh is small.
Figure~\ref{fig:efficiency} confirms this is the case -- the latency of \avh on 100K rows is 1.6 seconds on average, making this interactive. 
The figure further breaks down the time spends into three components:  (1) computing single-distribution metrics (green), (2) two-distribution metrics (blue), and (3) DQ programs (red), where computing two-distribution metrics (blue) takes the most time, which is expected. 
Overall, we see that the overall latency grows linearly with an increasing number of rows, indicating good scalability. 


\vspace{-2mm}
\section{Conclusions}

In this work, we develop an \avhfull (\avh) framework to automate data-validation in recurring pipelines. \avh can automatically generate explainable DQ programs that are provably accurate, by leveraging the statistical properties of the underlying metrics. Extensive evaluations on production pipelines show the efficiency and effectiveness of \avh.  

\clearpage

\break
\bibliographystyle{ACM-Reference-Format}
{\footnotesize
\bibliography{Auto-Validate-History}
}

\appendix

\iftoggle{fullversion}
{
\clearpage
    
\section{Details of statistical metrics}
\label{apx:detailed-statistics}

Table~\ref{tab:numeric-complete} and Table~\ref{tab:category-complete} give detailed descriptions of the statistical metrics used in \avh, which corresponds to simplified versions in Table~\ref{tab:numeric} and Table~\ref{tab:category}, respectively.

\begin{small}
\begin{table*}[t]
\vspace{-10mm}
    \centering
    \begin{tabular}{ccp{0.7\textwidth}}
    \toprule
    Type & Metric  & Description \\ 
    \midrule
    \multirow{5}*{Two-distribution} & $EMD$ & Earth Mover's distance (Wasserstein metric) \cite{rubner1998metric} between two numeric distributions \\
                        & $JS\_div$  & Jensen–Shannon divergence \cite{endres2003new} between two numeric distributions  \\
                        & $KL\_div$  & Kullback–Leibler divergence  (relative entropy)~\cite{endres2003new} between two numeric distributions \\
                        & $KS\_dist$  & Two-sample Kolmogorov–Smirnov test \cite{massey1951kolmogorov} between two numeric distribution (using p-value) \\
                        & $Cohen\_d$ & Cohen's d \cite{cohen2013statistical} that quantify the effect size between  two numeric distributions\\
    \hline
    \multirow{9}*{Single-distribution}  & $min$ & the minimum value  observed from a numeric column \\ 
                        & $max$  & the maximum value observed from a numeric column\\
                        & $mean$  & the arithmetic mean observed from a numeric column\\
                        & $median$  & the median value observed from a numeric column\\
                        & $sum$  & the sum of values  observed from a numeric column\\
                        & $range$  & the difference between max and min  observed from a numeric column\\
                        & $row\_count$ & the number of rows  observed from a numeric column\\
                        & $unique\_ratio$  & the fraction of unique values  observed from a numeric column\\
                        & $complete\_ratio$ & the fraction of complete (non-null) values  observed from a numeric column \\
    \bottomrule 
    \end{tabular}
    \caption{Statistical metrics used to generate DQ constraints for numerical data}
    \label{tab:numeric-complete}
\end{table*}
\end{small}

\begin{small}
\begin{table*}[t]
\vspace{-4mm}
    \centering
    \begin{tabular}{ccp{0.7\textwidth}}
    \toprule
    Metric Type & Metric & Description \\ 
    \midrule
    \multirow{12}*{Two-distribution} & $L_1$ & L-1  distance~\cite{cantrell2000modern} between two categorical distribution  \\
                        & $L_{inf}$ & L-infinity   distance~\cite{cantrell2000modern} between two categorical distributions\\
                        & $Cosine$ & Cosine distance \cite{gupta2021deep}{} between two categorical distributions\\
                        & $Chisquared$ & Chi-squared test \cite{fienberg1979use} using p-value between two categorical distributions\\
                        & $JS\_div$ & Jensen–Shannon divergence~\cite{endres2003new} between two categorical distributions\\
                        & $KL\_div$ & Kullback–Leibler divergence  (relative entropy)~\cite{endres2003new} between two categorical distribution\\   
                        & $Pat\_L_1$ & L-1 distance between the pattern profiles extracted from two categorical distributions \\
                        & $Pat\_L_{inf}$ & L-infinity distance between the pattern profiles extracted from two categorical distributions \\
                        & $Pat\_Cosine$ & Cosine distance between the pattern profiles extracted from two categorical distributions\\
                        & $Pat\_Chisquare$ & Chi-squared  p-value between the pattern profiles extracted from two categorical distributions \\
                        & $Pat\_JS\_div$ &  Jensen–Shannon divergence of the pattern profiles extracted from two categorical distributions \\
                        & $Pat\_kl\_div$ & Kullback–Leibler divergence of the pattern profiles extracted from two categorical distributions\\
    \hline
    \multirow{8}*{Single-distribution}  & $str\_len$ & the average length of strings observed in a categorical column\\
                        & $char\_len$ & the average string length for values observed in a categorical column\\
                        & $digit\_len$ & the average number of digits in values observed from a categorical column\\
                        & $punc\_len$ & the average number of punctuation in values observed from a categorical column\\
                        & $row\_count$ & the number of rows  observed from a categorical column\\
                       & $unique\_ratio$ & the fraction of unique values  observed from a categorical column\\
                        & $complete\_ratio$ & the fraction of complete (non-null) values  observed from a categorical column\\
                        & $dist\_val\_count$ & the number of distinct values  observed from a categorical column\\
    \bottomrule 
    \end{tabular}
    \caption{Statistical metrics used to generate DQ constraints for categorical data. 
    }
    \vspace{-4mm}
    \label{tab:category-complete}
\end{table*}
\end{small}

\section{Synthetic ``training'' data}
\label{apx:datagen}

We carefully reviewed the DQ literature and cataloged a list of 10 common types of DQ issues in pipelines, so that we can systematically synthesize data deviations that are due to DQ issues, which would help us to select the most salient ``features'' or DQ constraints that are sensitive in detecting common DQ deviations. 

We enumerate the list of $10$ different types of DQ issues below, as well as the parameters we use (to control deviations with different magnitudes). By injecting varying amounts of DQ issues into a given column $C$, we generates a total of 60  variations $C'$ for each $C$ (e.g., different fractions of values in $C$ are replaced with nulls for the type of DQ issue ``increased nulls''). Collectively, we denote this set of synthetically generated DQ issues on $C$ as $\mathbf{D}(C)$.

\textbf{DQ Issue Type 1: Schema change. }  We replace $p\%$ (with $p=1, 10, 100$) of values in a target column $C$ for which we want to inject DQ variation, using values randomly sampled from a neighboring column of the same type. This is to simulate a ``schema change'', where some fraction of values in a different column are either partially mis-aligned (e.g., due to a missing delimiter or bad parsing logic), or completely mis-aligned (e.g., due to extra or missing columns upstream introduced over time). Note that $p=100$ corresponds to a complete schema-change, otherwise it is a partial schema change. 

\textbf{DQ Issue Type 2: Change of unit.} To simulate a change in the unit of measurement, which is a common DQ issue (e.g., reported by Google in~\cite{polyzotis2017data, breck2019data} like discussed earlier), we synthetically multiply values in a numeric column $C$ by x10, x100 and x1000. 

\textbf{DQ Issue Type 3: Casing change.} To simulate possible change of code-standards (e.g., lowercase country-code to uppercase, as reported by Amazon~\cite{schelter2018deequ}), we  synthetically change $p\%$ fraction of values  (with $p=1, 10, 100$) in $C$, from lowercase to uppercase, and vice versa.

\textbf{DQ Issue Type 4: Increased nulls. } Since it is a sudden increase of null values such as NULL/empty-string/0 is  common DQ issue, we sample $p\%$ values in a $C$ (with $p=1, 50, 100$), and replace them with empty-strings in the case of categorical attribute, and 0s in the case of numerical attribute.

\textbf{DQ Issue Type 5: Change of data volume.} Since a sudden increase/decrease of row counts can also be indicative of DQ issues~\cite{polyzotis2017data}, we up-sample values in $C$ by a factor of x2, x10, or down-sample $C$ with only 50\%, 10\% of the values.

\textbf{DQ Issue Type 6: Change of data distributions}. To simulate a sudden change of data distributions~\cite{breck2019data}, we sorted all values in $C$ first, then pick the first or last $p\%$ values as a biased sample and replace $C$, with $p=10, 50$.

\textbf{DQ Issue Type 7: Misspelled values by character perturbation.} Typos and misspellings is another type of common DQ issue (e.g., ``Missisipii'' and ``Mississippi''), frequently introduced by humans when manually entering data. To simulate this type of DQ issue, we randomly perturb $p\%$ of characters in $C$ to a different character of the same type (e.g., $[0-9] \rightarrow [0-9]$, and $[a-z] \rightarrow [a-z]$), with $p=1, 10, 100$. 

\textbf{DQ Issue Type 8:  Extraneous values by character insertion.}  Sometimes certain values in a column $C$ may be associated with extraneous characters that are not expected in clean data. To simulate this, for each value in $C$, we insert randomly generated characters with  probability $p\%$, where $p=10, 50$.

\textbf{DQ Issue Type 9: Partially missing values by character deletion. } Sometimes certain values in a column $C$ may get partially truncated, due to issues in upstream logic. We simulate this by deleting characters for values in $C$ with probability  $p\%$, where $p=10, 50$. 

\textbf{DQ Issue Type 10: Extra white-spaces by padding.} We randomly insert leading or tailing whitespace for $p\%$ of values, where $p=10, 50, 100$. 

While we are clearly not the first people to report these aforementioned DQ issues, we are the first to systematically catalog them and synthetically generate such DQ variations, and are the first to use them as ``training data'' that guides a DQ algorithm to select the most salient DQ features specific to the characteristics of a column $C$. We release our generation procedures in~\cite{Full}, which can be used for future research.


\section{Pattern generation}
\label{appendix:pattern}

In addition to use raw values from columns and compare their distributional similarity (e.g., using $L_1$, $L_{inf}$, $Cosine$, etc.),  sometimes values in a column follows a specific pattern, for example, timestamp values like "2022-03-01 (Monday)", currency values like ``\$19.99'', zip-codes like ``98052-1202'', etc.
For such values, comparing distributions for raw values that are drawn from a large underlying domain induced by patterns (e.g., time-stamp), typically yields very small overlap/similarity because of the large space of possible values in the underlying domain. (This is in contrast to small categorical domains with a small number of possible values, where distributional similarity is usually high and more meaningful).

In \avh, we observe that the pattern strings for such pattern-induced domains is an orthogonal representation of values in a column, which gives another way to ``describe'' the column and detect possible DQ deviations. For example, timestamp values like "2022-03-01 (Monday)" can be generalized to a pattern "$\backslash d$$\backslash d$$\backslash d$$\backslash d$-$\backslash d$$\backslash d$-$\backslash d$$\backslash d$ ($\backslash l\backslash l\backslash l\backslash l\backslash l\backslash l$)", currency values like ``\$19.99'' can be generalized to ``\$$\backslash d$$\backslash d$.$\backslash d$$\backslash d$'', etc. Assuming that the format of the data is changed due to upstream DQ issue, e.g., currency values become mixed where some values have no currency-signs, or time-stamps becomes mixed with multiple formats of time-stamps, a distributional similarity of the pattern strings above provides a powerful way to ``describe'' the expected pattern distribution in a column, which makes it possible to catch DQ issues in columns whose underlying domains are pattern-related.  

For the metrics that have a prefix ``Pat\_'' in Table~\ref{tab:category}, we first generate pattern-strings for each value $v\in C$, by converting each character in $v$ to a wildcard character following a standard [0-9] $\rightarrow$ $\backslash d$ (for digits), [a-zA-Z] $\rightarrow$ $\backslash l$ (for letters), and replace all punctuation as ``-''. We then compute the same distributional similarity (e.g., $L_1$, $L_{inf}$, $Cosine$, etc.), just as regular distributional similarity metrics for raw string values. (Note that \avh is robust to a large space of DQ constraints, and can intelligently select the most salient features, such that for columns where pattern-based DQ is not a good DQ description, such pattern-based DQ constraints will not be selected automatically.)


\begin{small}
\begin{algorithm}[t]
\SetKwInOut{Input}{Input}
\SetKwInOut{Output}{Output}
\SetKw{KwReturn}{return}

\Input{A categorical column $C$}
\Output{The column pattern $C'$}

\ForEach{$value \in C$}{
    \ForEach{ \text{same pattern of consecutive} $chars \in value$}{
        \If{$char \in [0-9]$} {
            Replace consecutive chars with a symbol \bm{$\backslash d$}
        }
        \If{$char \in [A-Z] \  or \  [a-z]$} {
            Replace consecutive chars with a symbol \bm{$\backslash l$}
        }
        \If{$char \in [punctuation]$} {
            Replace consecutive chars with a symbol \bm{$-$}
        }
    }
}
$C' \leftarrow C$ \\
\KwReturn $C'$ 
\caption{Pattern Generation}
\label{algo:pattern}
\end{algorithm}
\end{small}

\section{construct constraints: Pseudocode}
\label{apx:pseudo-code}

We show the pseudo code to construct DQ constraints in Algorithm~\ref{algo:constraints}. This procedure directly corresponds to Section~\ref{sec:dq_constraint}

\begin{algorithm}[t]
\SetKwInOut{Input}{Input}
\SetKwInOut{Output}{Output}
\SetKw{KwBy}{by}
\SetKw{kwReturn}{return}
 \Input{Metrics $\mathbf{M}$, history $H = \{C_1, C_2, \ldots\}$ of col $C$}
 \Output{Constructed DQ constraints $\mathbf{Q}$}
 $\mathbf{Q} \leftarrow \emptyset$
 
\ForEach{$M \in \mathbf{M}$}
{
    $M(H) \leftarrow \{M(C_1), M(C_2), \ldots, M(C_K)\}$ 
    
    $M(H) \leftarrow$ process-stationary($M(H)$)  \tcp*[h]{Algorithm~\ref{algo:stationary}} \label{line:subprocedure}

    $\mu \leftarrow $ mean of $M(H)$, $\sigma^2 \leftarrow$ variance of $M(H)$
    
  \ForEach{$\beta \in [\sigma, n \sigma],$ \text{increasing with a step-size} $s$,}
    {
    $Q_i \leftarrow Q(M, C, \mu - \beta, \mu + \beta)$ 
    
    $FPR(Q_i) \leftarrow \text{calc-FPR}(M, \beta)$ \tcp{\small{Equation~\eqref{eqn:chebyshev}-\eqref{eqn:clt}}}
    
    $\mathbf{Q} \leftarrow \mathbf{Q} \cup Q_i$
    }
}

\kwReturn $\mathbf{Q}$

\caption{Construct DQ constraints $\mathbf{Q}$}
\label{algo:constraints}
\end{algorithm}

\begin{algorithm}[t]
\SetKwInOut{Input}{Input}
\SetKwInOut{Output}{Output}
\SetKw{KwReturn}{return}

 \Input{$M(H) = $ $\{M(C_1), M(C_2), \ldots, M(C_K)\}$}
 \Output{Processed $M'(H)$ that is stationary 
 }

$is\_stationary \leftarrow ADF(M(H))$ \tcp*{Perform ADF test}

\If{$is\_stationary$} 
    {
    \KwReturn $M(H)$ 
    }
\Else
    {
    $M'(H) \leftarrow$ time-series-differencing($M(H)$) \tcp{{using first-order and seasonal differencing}}
    
    \KwReturn $M'(H)$ 
    }
\caption{Time-series differencing for stationary}
\label{algo:stationary}
\end{algorithm}

\section{Time-series Differencing}
\label{appendix:differencing}

Algorithm~\ref{algo:stationary} gives an overview of the time-series differencing step, which we will expand and explain in this section. Details of this step can be found in Algorithm \ref{algo:stationary-process}.

\begin{small}
\begin{algorithm}[h]
\SetKwInOut{Input}{Input}
\SetKwInOut{Output}{Output}
\SetKw{KwReturn}{return}

 \Input{$M(H) = $ $\{M(C_1), M(C_2), \ldots, M(C_K)\}$}
 \Output{Processed stationary $M'(H)$ 
 }

$is\_stationary \leftarrow ADF(M(H))$ \tcp*{Perform ADF test}

\If{$is\_stationary$} {
    \KwReturn $M(H)$ 
}
\Else{
    \text{// Perform lag transformation without log transformation}
    \ForEach{$lag \in [1, K-1]$} {
        $M'(H) \leftarrow lag\_transform(M(H),\  lag$) \\
        $is\_stationary \leftarrow ADF(M'(H))$ \\
        \If{$is\_stationary$}{
            \KwReturn $M'(H)$
        }
    }
   
    \text{// Perform lag transformation with log transformation}
    \ForEach{$lag \in [1, K-1]$} {
        $M'(H) \leftarrow lag\_transform\_with\_log(M(H),\  lag$) \\
        $is\_stationary \leftarrow ADF(M'(H))$ \\
        \If{$is\_stationary$}{
            \KwReturn $M'(H)$
        }
    }
    
\KwReturn $None$
}
\caption{Time-series differencing for stationary (Details)}
\label{algo:stationary-process}
\end{algorithm}
\end{small}

Recall that time-series differencing aims to make time series stationary with static underlying parameters. We start by performing the ADF test to determine the stationarity of $M(H)$, and return $M(H)$ if it is already stationary. If it is not, we then perform lag-based transforms~\cite{granger1980introduction}. Given a sequence of $M(H) = \{M(C_1), M(C_2), \ldots, M(C_K)\}$, a lag-based transform with $lag=l$ is defined as 
\begin{small}
$$M(H)^{lag=l} = \{M(C_{l+1}) - M(C_{1}), M(C_{l+2}) - M(C_{2}), \ldots, M(C_{l+K} - M(C_{l}))\}$$ 
\end{small}
Which performs a difference step for two events that are $l$ time-steps away. Observe that such a differencing step handles cyclic data with periodic patterns (e.g., weekly user traffic data can be differenced away with $lag=7$), as shown in Example~\ref{ex:differencing} earlier.  

For each $lag \in [1, K-1]$, if the resulting $M(H)^{lag}$ is already stationary (passes the ADF test), we return the corresponding $M(H)^{lag}$ for the next stage for \avh to auto-program DQ (and remember the $lag$ parameter to pre-process data arriving in the future).

If none of the $lag$ parameter leads to a stationary time-series, we additionally perform a log transform on $M(H)$, which can better handle time-series with values that are orders of magnitude different. We repeat the same process as lag-only transforms like above, until we find a stationary time-series or we return None (in which case, the sequence $M(H)$ associated with this metric $M$ will be ignored by downstream \avh due to its non-stationary nature. Also note that it is possible to perform additional second-order or third-order differencing, which we omit here).

\section{Proof of Proposition 2}
\label{appendix:proof_cantelli}

\begin{proofsketch}
We prove this proposition using Cantelli's inequality~\cite{bulmer1979principles}. Cantelli's inequality states that for a random variable $X$, there is a class of one-sided inequality in the form of $P(X - \mu \geq k \sigma) \leq \frac{1}{1 + k^2}$, $\forall k \in \mathbb{R}^+$. 

For metrics $M \in \{EMD, JS\_div, KL\_div,$ $KS\_dist,$ $Cohen\_d,$ $L_1,$ $L_{inf},$ $Cosine,$ $Chisquared\}$, which are ``distance-like'' metrics, DQ constraints can be one-sided only, to guard against deviations with distances larger than usual, e.g., newly-arrived data whose distance from previous batches of data are substantially larger than is typically expected. (On the other hand, if the distance of new data and previous batches of data are smaller than usual, this shows more homogeneity and is typically not a source of concern). 

For this reason, we can apply Cantelli's inequality for metrics $M \in \{EMD, JS\_div, KL\_div,$ $KS\_dist, Cohen\_d, L_1,$ $L_{inf},$ $Cosine,$ $Chisquared\}$, with one-sided DQ. For such a metric $M$, let $M(C)$ be our random variable. Let $k = \frac{\beta}{\sigma}$. Replacing $k$ with $\frac{\beta}{\sigma}$ above, we get $P(M(C) - \mu \geq \beta) \leq \frac{\sigma^2}{\sigma^2 + \beta^2}$. 
Note that $P(M(C) - \mu \leq \beta)$ is exactly our one-sided DQ for metrics with distance-like properties.  We thus get $P(Q \text{ violated on } C)$ $\leq \frac{\sigma^2}{\sigma^2 + \beta^2}$, which is equivalent to saying that the expected FPR of $Q$ is no greater than $\frac{\sigma^2}{\sigma^2 + \beta^2}$.
 \end{proofsketch}

\section{Proof of Proposition 3}
\label{appendix:proof_clt}

\begin{proofsketch}
We prove this proposition using Central Limit Theorem (CLT)~\cite{bulmer1979principles}. 
Recall CLT states that when  independent random variables are summed up and normalized, it tends toward normal distribution. Metrics $M \in \{count, mean,$ $str\_len,$ $char\_len, digit\_len,$ $punc\_len, complete\_ratio\}$, can all be viewed as the sum of independent random variables (for example, $str\_len,$ $char\_len,$ $digit\_len$ etc. are straightforward sum of these functions applied on individual cells; $count$ are the 0/1 sum for a random variable indicating tuple presence/not-presence, etc.). Such sums are then averaged over all cells in the same column $C$, which would tend to normal distributions per CLT. We can thus apply the tail bound of normal distributions, making it possible to apply tail bounds of normal distributions. 

For $M \in \{count, mean,$ $str\_len,$ $char\_len, digit\_len,$ $punc\_len,$ $complete\_ratio\}$, let $M(C)$ be our random variable. From tail bounds of normal distributions, we know  $P(-k \sigma \leq M(C) - \mu \leq k \sigma) = erf(\frac{k}{\sqrt{2}})$~\cite{normal-error-function-66-95-99-rule}, where $erf(x)$ is the Gauss error function. Let $k = \frac{\beta}{\sigma}$. Replacing $k$ with $\frac{\beta}{\sigma}$ above, we get $P(-\beta \leq M(C) - \mu \leq \beta) = erf(\frac{\beta}{ \sqrt{2} \sigma})$ $=  \frac{2}{\sqrt{\pi}}\int_{0}^{\frac{\beta}{\sqrt{2}\sigma}} e^{-t^2} dt$. Note that $P(-\beta \leq M(C) - \mu \leq \beta)$ is exactly $P(Q \text{~satisified on C})$, thus we get $\mathrm{E}[FPR(Q)] = 1 - \frac{2}{\sqrt{\pi}}\int_{0}^{\frac{\beta}{\sqrt{2}\sigma}} e^{-t^2} dt$.
\end{proofsketch}

\section{Hardness of the \avh problem}
\label{appendix:hardness}

\begin{proposition}
The \avh problem in Equation~\eqref{eqn:recall-2}-Equation~\eqref{eqn:conjunctive-2} is NP-hard. Furthermore, it cannot be approximated within a factor of $(1-\frac{1}{e})$ under standard assumptions.
\end{proposition}

\begin{proofsketch}
We show the hardness using a reduction from the Maximum Coverage problem~\cite{nemhauser1978analysis}. Recall that in Maximum Coverage, we are given a set of sets $S$, and the objective is to find a subset $S' \subseteq S$ such that the union of the elements covered by $S'$, $|\bigcup_{S_i \in S'}{S_i}|$, is maximized, subject to a cardinality constraint $|S'| \leq K$.

We show a polynomial time reduction from Maximum Coverage to \avh as follows. For any instance of Maximum Coverage with $S = \{S_i\}$, we construct the an \avh problem by converting each $S_i$ into a DQ constraint $Q_i$, whose $FPR(Q_i)$ is unit cost 1, and recall $R(Q_i)$ is exactly the set of elements in $S_i$. If we could solve the corresponding \avh problem in polynomial-time, we would have solved the Maximum Coverage, thus contracting the hardness of Maximum Coverage.

Also note that through the construction above, the objective value of Maximum Coverage is identical to that of \avh. Thus we can use the inapproximation results from Maximum Coverage ~\cite{nemhauser1978analysis}, to show that \avh cannot be approximated within a factor of $(1 - \frac{1}{e})$.
\end{proofsketch}

\section{Proof of Proposition 4}
\label{appendix:approx}

\begin{proofsketch}
We show that Algorithm~\ref{algo:avh} is a $(\frac{1}{2} - \frac{1}{2e})$ approximation algorithm for the \avh problem, which follows from the Budgeted Maximum Coverage problem~\cite{khuller1999budgeted}. Recall that in Budgeted Maximum Coverage problem, we are given a set of sets $S = \{S_i\}$, where each set $S_i$ has a cost $c(S_i)$, and each element in sets has a weight $w(e_j)$, the objective is to find a subset $S' \subseteq S$ such that the weight of all elements covered by $S'$ is maximized, subject to a budget constraint $\sum_{S_i \in S'}c(S_i) \leq B$. We show that for any instance of our \avh problem, it can be converted to Budgeted Maximum Coverage  as follows. We convert each $Q_i$ into a set $S_i$, and let the cost $c(S_i)$ be $FPR(Q_i)$. Furthermore, we convert the set of recall items into elements in Budgeted Maximum Coverage, and set the weight of each element to unit weight. Finally, we let the elements covered by $S_i$ in  Budgeted Maximum Coverage  to be exactly the $R(Q_i)$ in \avh. The approximation ratio in Proposition~\ref{prop:approx} follows directly from the Theorem 3 of~\cite{khuller1999budgeted} now.

We note that there are an alternative algorithm with better approximation ratio $(1 - \frac{1}{e})$~\cite{khuller1999budgeted}, which however is of complexity $|\mathbf{Q}|^3$, where $|\mathbf{Q}|$ is the number of DQ constraints constructed from Algorithm~\ref{algo:constraints}. Because $|\mathbf{Q}|$ is at least in the hundreds, making the alternative very expensive in practice and not used in our system.

We also show that the solution $S$ from our Algorithm~\ref{algo:avh} is a feasible solution of \avh, whose expected FPR is lower than $\delta$. In order to see this, recall that we construct DQ $Q_i \in \mathbf{Q}$ and estimate each $Q_i$'s worst case $FPR(Q_i)$ following Proposition~\ref{prop:chebyshev},~\ref{prop:cantelli},~\ref{prop:clt}. Algorithm~\ref{algo:avh} ensures that $\sum_{Q_i \in S}{FPR(Q_i)} \leq \delta$. For the conjunctive program $P(S)$ induced by $S$, the FPR of  $P(S)$ follows the inequality $FPR(P(S)) \leq \sum_{Q_i \in S}{FPR(Q_i)}$, because the false-positives from  $P(S)$, is produced by a union of the false-positives from each $Q_i \in S$. Combining this with $\sum_{Q_i \in S}{FPR(Q_i)} \leq \delta$, we get $FPR(P(S)) \leq \delta$.
\end{proofsketch}

\section{Sensitivity Analysis}
\label{apx:exp-sensitivity}

\begin{figure*}[t]
\vspace{-12mm}
\begin{minipage}[c]{0.3\linewidth}
    \centering
    \includegraphics[height=4cm]{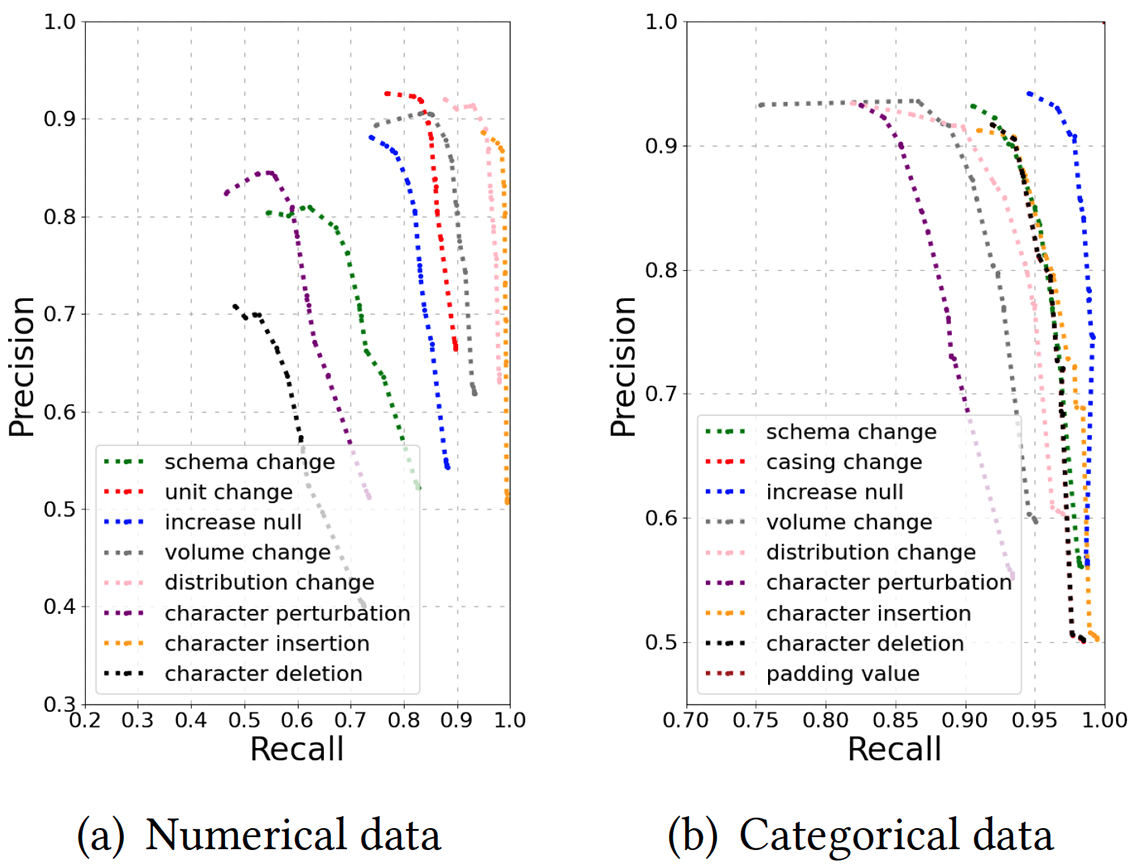} 
    \vspace{-2mm}
    \caption{Results by DQ types}
    \label{fig:syn-by-types}
\end{minipage}
\hspace{3mm}
\begin{minipage}[c]{0.3\linewidth}
    \centering
    \includegraphics[height=4cm]{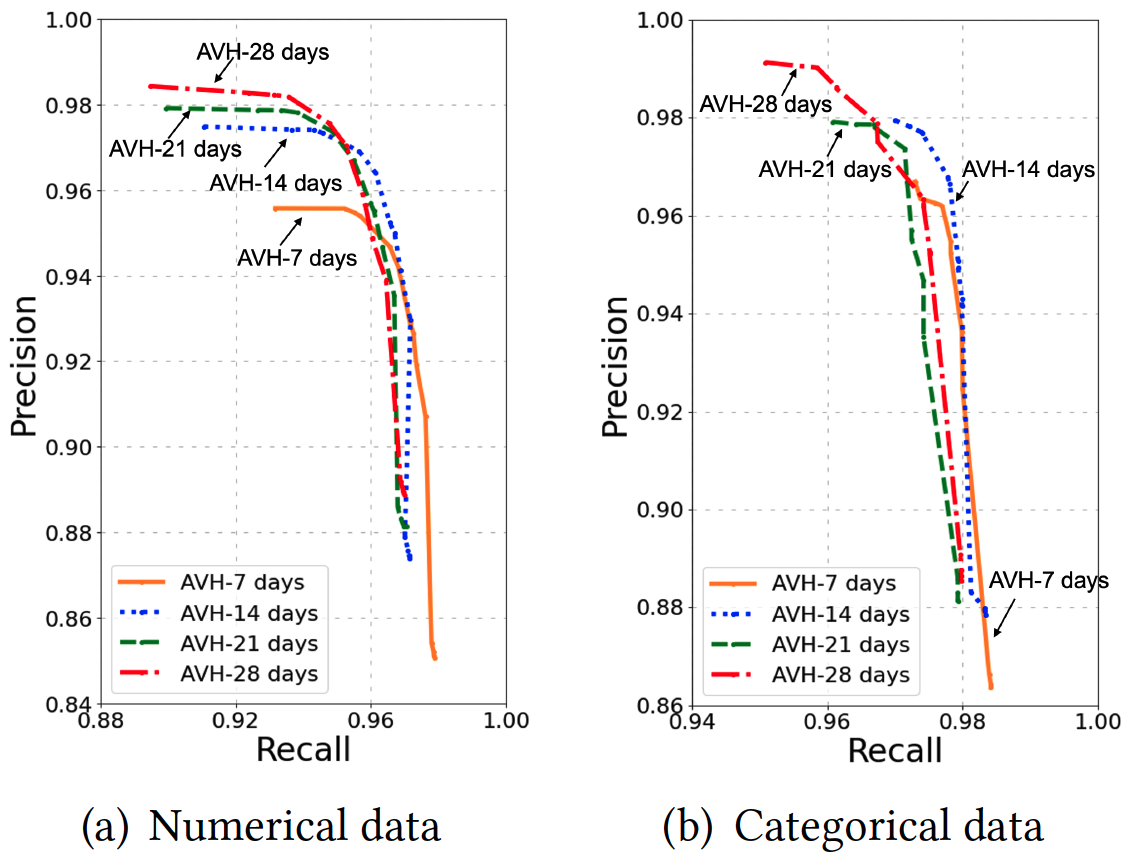} 
    \vspace{-2mm}
    \caption{Sensitivity to history length}
    \label{fig:vary day}
\end{minipage}
\hspace{3mm}
\begin{minipage}[c]{0.3\linewidth}
    \centering
    \includegraphics[height=4cm]{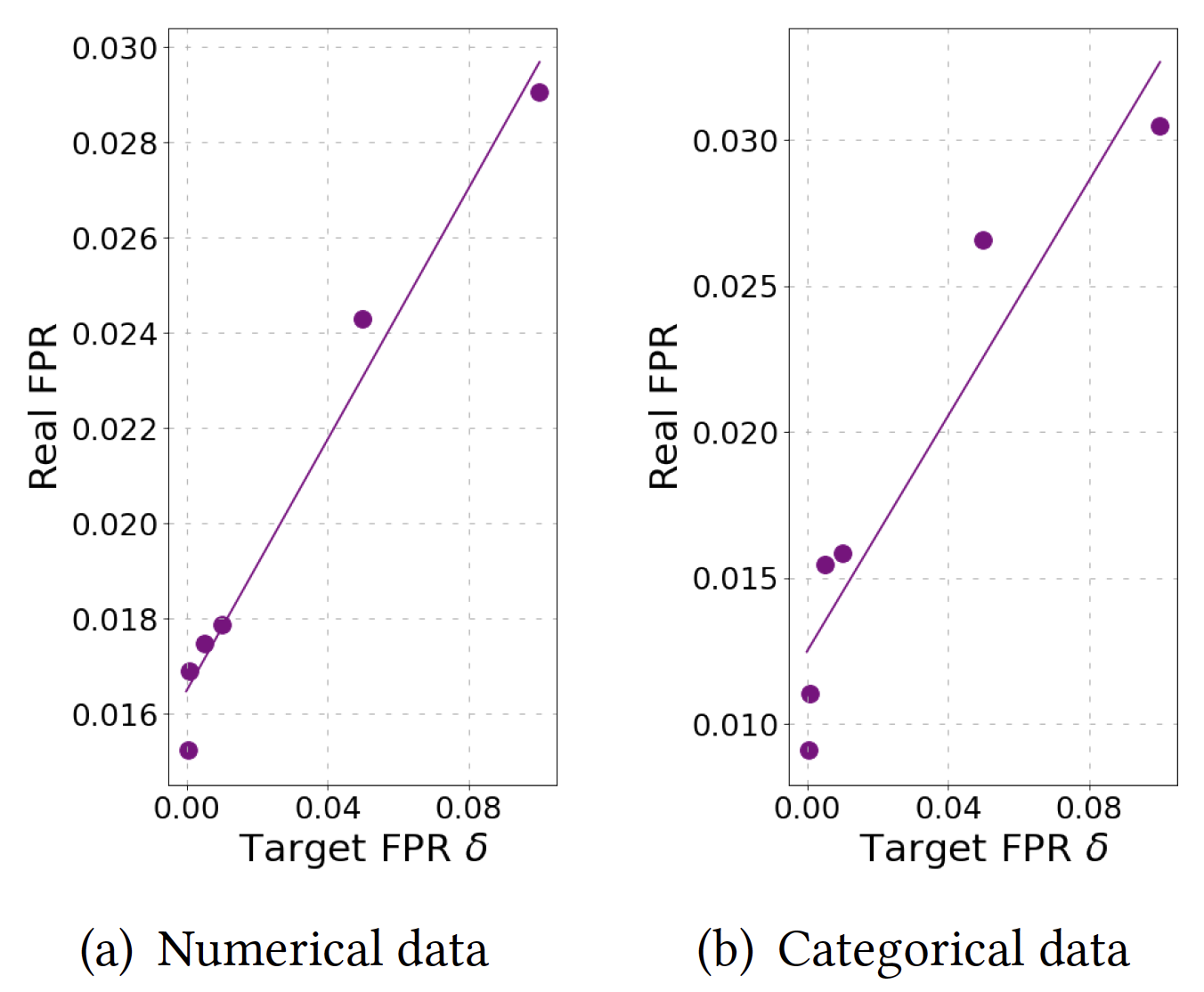} 
    \vspace{-2mm}
    \caption{Sensitivity to target FPR $\delta$}
    \label{fig:fpr}
\end{minipage}
\end{figure*}

We perform extensive experiments to understand the sensitivity of our method.
\iftoggle{fullversion}
{
}
{
    We discuss a subset of our sensitivity and ablation studies in this section (additional results can be found in~\cite{full}).
}

\underline{Sensitivity to history length.}
Since \avh leverages a history of past pipeline executions, 
where the number of past executions likely has an impact on accuracy. Figure \ref{fig:vary day} shows the accuracy results for numerical and categorical data respectively, when \{7, 14, 21, 28\} days of historical data are available.
Overall, having 28-day history leads to the best precision, though with 14 and 7-day histories \avh also produces competitive results. We highlight that unlike traditional ML methods that typically require more than tens of data points (Figure~\ref{fig:real} suggests that even 30-day history is not sufficient for ML methods), \avh exploits the unique statistical properties of the underlying metrics (e.g., Chebyshev and CLT), and can work well even with limited data, which is a unique characteristic of \avh.




\underline{Sensitivity to target precision $\delta$.}
Figure \ref{fig:fpr} shows the relationship between the target FPR $\delta$ parameter used in \avh (Equation~\eqref{eqn:fpr}), and the real FPR observed on \avh results (we note that 1-FPR corresponds to the precision metric). On both numerical and categorical data, the real FPR increases slightly when a larger target FPR $\delta$ is used, showing the effectiveness of this knob $\delta$ in \avh. Also note that the real FPR is consistently lower than the target-FPR, likely due to the conservative nature of the statistical guarantees we leverage (Chebyshev and Cantelli's inequalities we use in Proposition~\ref{prop:chebyshev} and~\ref{prop:cantelli} give worst-case guarantees).

\section{Ablation Studies} 
\label{apx:exp-ablation}

We perform additional ablation studies, to understand the importance of different components used in \avh. 

\underline{Effect of using single/two-distribution metrics.}
Recall that in \avh, we exploit both single-distribution and two-distribution metrics (in Table~\ref{tab:numeric} and Table~\ref{tab:category}) to construct DQ programs. A key difference of the two types of metrics, is that computing two-distribution metrics (e.g., $L_{inf}$ and $L_1$) would require both the current column $C_K$ and its previous snapshot $C_{K-1}$ (e.g., in $L_{1}(C_K, C_{K-1})$). This requires raw data from the previous run $C_{K-1}$ to be kept around, which can be costly in production big-data systems. In contrast, single-distribution metrics (e.g., $row\_count$ and $unique\_ratio$) can be computed on $C_K$ and $C_{K-1}$ separately, and we only need to keep the corresponding metrics from $C_{K-1}$ without needing to keep the raw $C_{K-1}$, which makes single-distribution metrics a lot more efficient and inexpensive to use in \avh.

\begin{figure*}[t!]
\begin{minipage}[c]{0.3\linewidth}
    \centering
    \includegraphics[height=4cm]{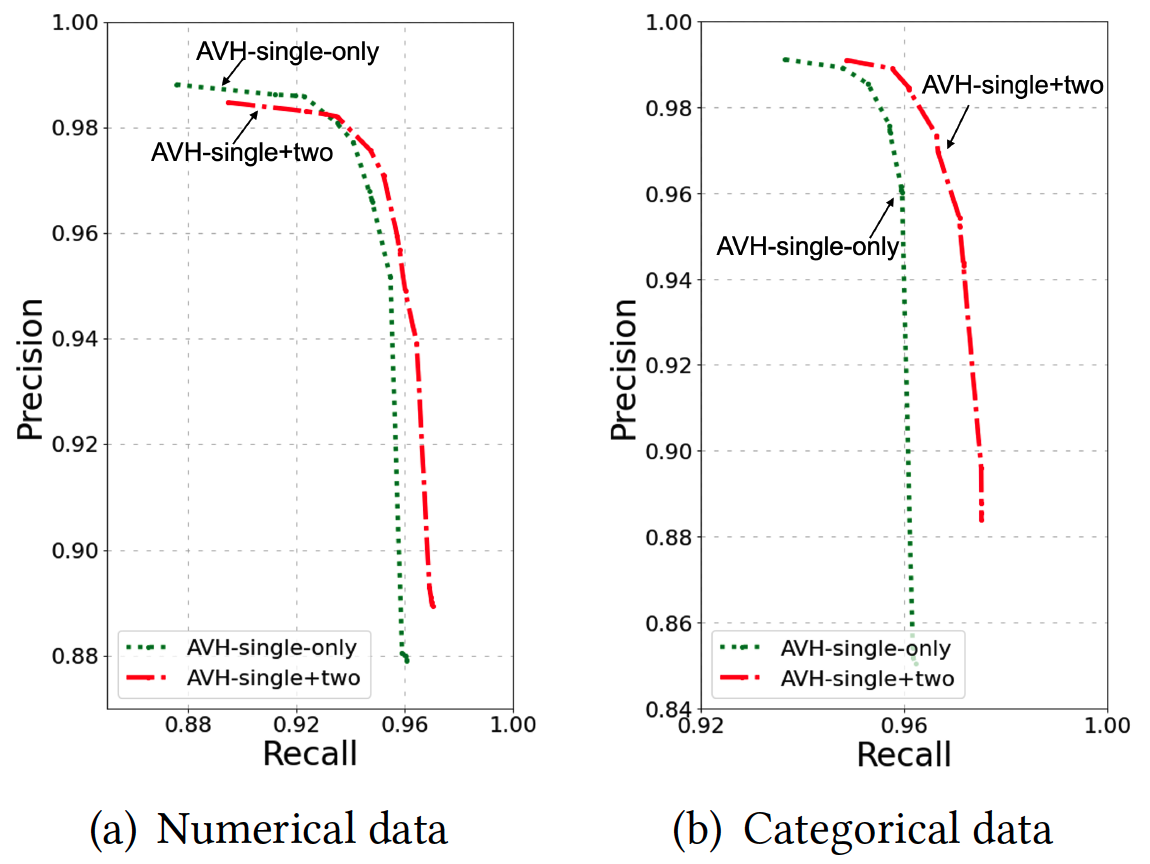} 
    \vspace{-2mm}
    \caption{Effect of two-dist. metrics}
    \vspace{-2mm}
    \label{fig:single-two-comp}
\end{minipage}
\hspace{3mm}
\begin{minipage}[c]{0.3\linewidth}
    \centering
    \includegraphics[height=4cm]{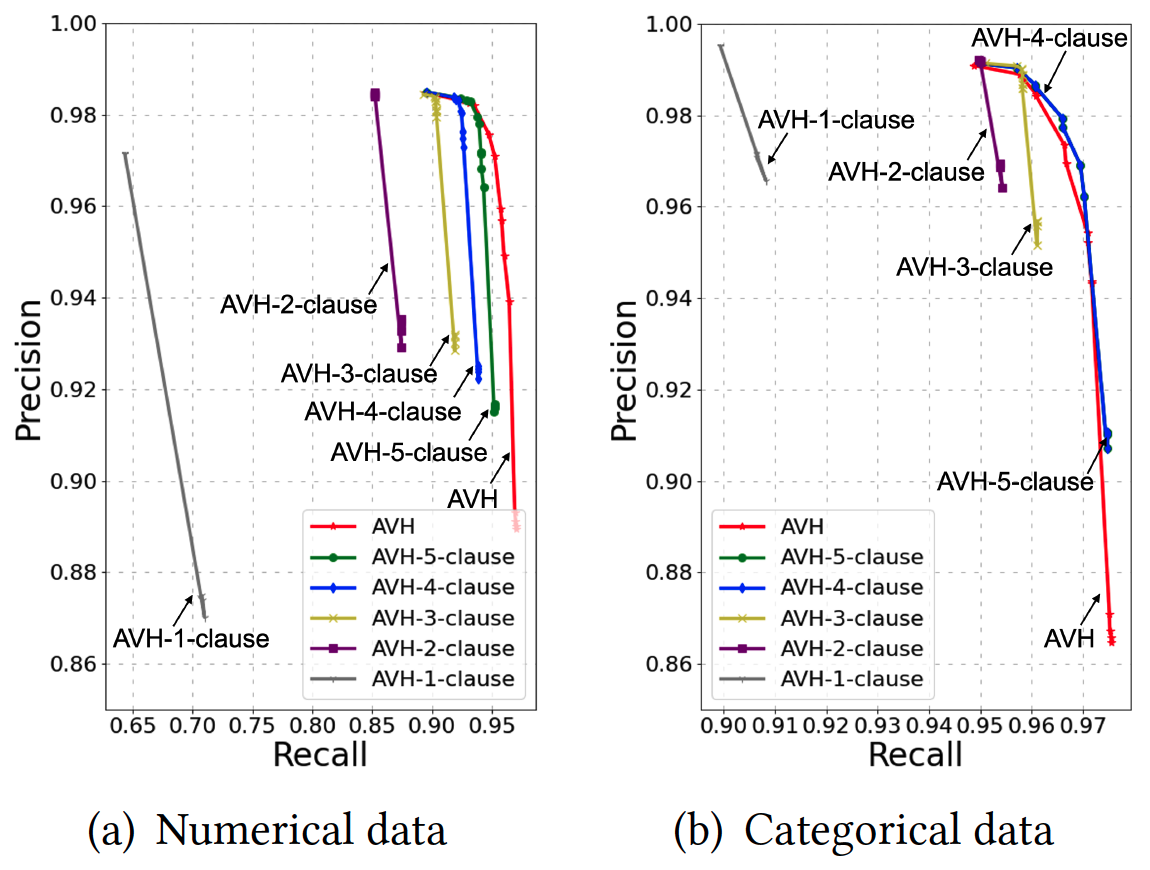} 
    \vspace{-2mm}
    \caption{Effect of clause count limit}
    \vspace{-2mm}
    \label{fig:clause limit}
\end{minipage}
\hspace{3mm}
\begin{minipage}[c]{0.3\linewidth}
    \centering
    \includegraphics[height=4cm]{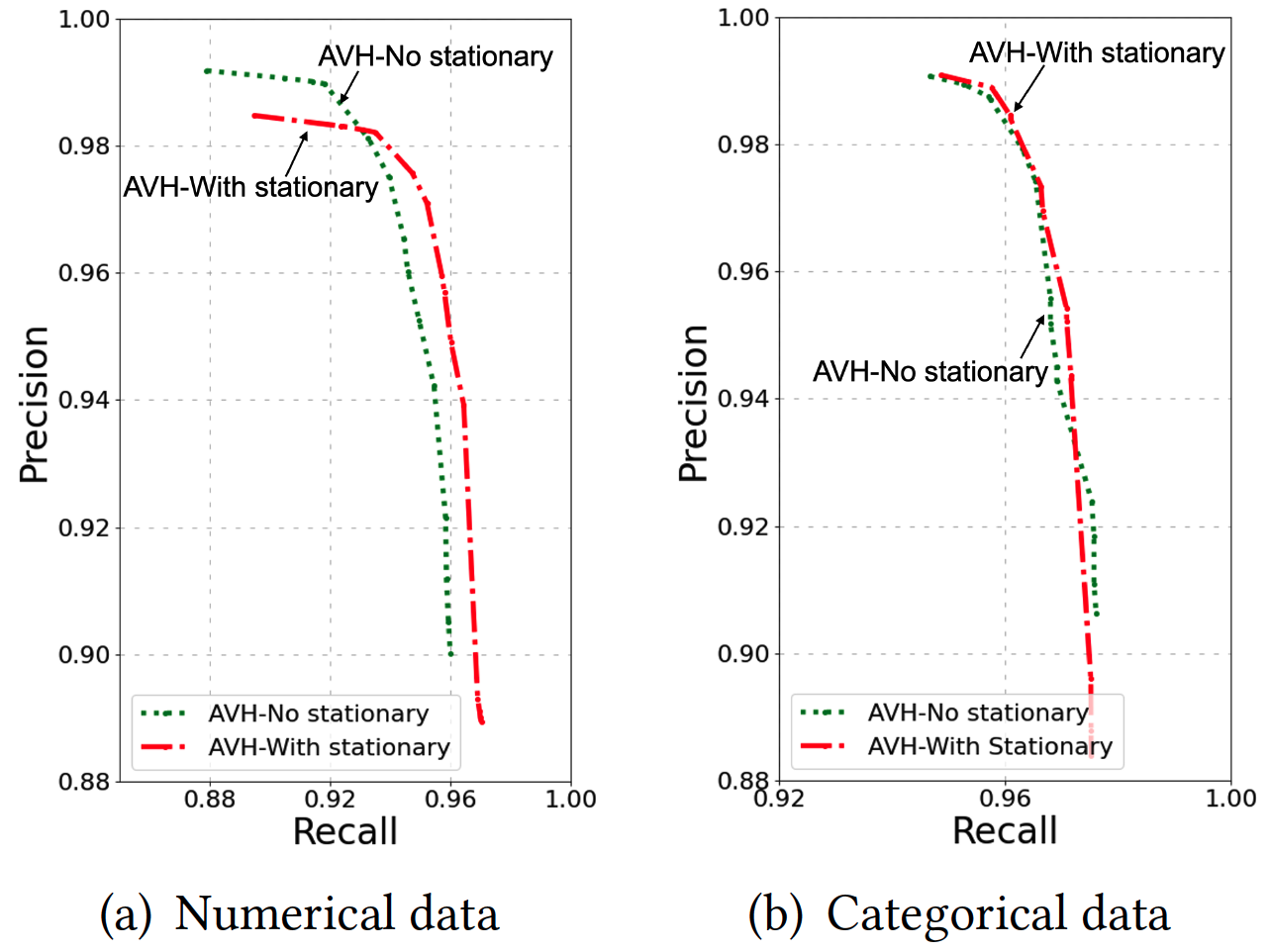}
    \vspace{-2mm}
    \caption{Using stationary processing}
    \vspace{-2mm}
    \label{fig:stationary}
\end{minipage}
\end{figure*}

In Figure \ref{fig:single-two-comp}, we compare the full \avh (with both single- and two-distribution metrics), with \avh using only single-distribution metrics.  Encouragingly, the latter variant produces comparable quality with the full \avh, likely because the large space of single-distribution metrics is already rich and expressive enough. This suggests that we can deploy \avh inexpensively without using two-distribution metrics, while still reaping most of the benefits.

\underline{Effect of limiting the number of DQ clauses.} We also study the number of DQ clauses that \avh generates, because intuitively, the more clauses it generates, the more expressive the DQ programs become, at the cost of human explainability/interpretability (there is a setting in \avh that engineers can review and approve auto-suggested DQ programs). We report that on numerical data in the \texttt{Real} benchmark, the median/mean of the number of clauses \avh generates is 3 and 2.62, respectively; on categorical data the median/mean is 2 and 2.15, respectively.
We believe this shows that the programs \avh generates are not only effective but also simple/understandable. In Figure~\ref{fig:clause limit} we impose an artificial limit on the number of clauses that \avh can generate (in case better readability is required). We observe a drop in performance when only 1 or 2 clauses are allowed, but the performance hit becomes less significant if we allow 3 clauses.

\underline{Effect of stationary processing.} Since we use stationarity test and stationary processing for statistics that are time-series (Section~\ref{sec:dq_constraint}), in Figure~\ref{fig:stationary} we study its effect on overall quality. We can see that for numerical data, stationary processing produces  a noticeable improvement, which however is less significant on categorical data.

    \section{Manual Review of Pipeline Data} 
    \label{sec:error}
    
    Based on our conversations with data engineers and data owners, the data tables collected from production data pipelines used in our \texttt{Real} benchmark (describe in Section~\ref{sec:benchmark}) are production-quality and likely free of DQ issues, because these files are of high business impact with many downstream dependencies, such that if they had any DQ issues they would have already been flagged and fixed by data engineers.
    
    In order to be sure, we randomly sampled 50 categorical and 50 numerical data columns, and manually inspected these sample data in the context of their original data tables, across 60 snapshots, to confirm the quality of the benchmark data. We did not find any DQ issues based on our manual inspection. We also perform a hypothesis test based on the manual analysis, with $H_0$ stating that over 3\% of data has DQ issues. Our inspection above  rejects the null hypothesis (p-level=0.05), indicating that it is highly unlikely that the benchmark data has DQ issues, which is consistent with the assessment from data owners, and confirms the quality of the data used in the \texttt{Real} benchmark.
    
    During our conversations with data engineers, we were pointed to three known DQ incidents, which we collect and use as test cases to study \avh's coverage. \avh was able to detect all such known DQ cases based on historical data.
    Figure~\ref{fig:error-fp} shows such an example that is intuitive to see. Here each file is an output table (in csv format) produced by a daily recurring pipeline. As can be seen in the figure, for the file produced on ``2019-01-19'', the file size (and thus row-count) is much larger than the days before and after ``2019-01-19'' (21KB vs. 4KB). While small in scale, we believe this study on user-provided data further confirms the effectiveness of \avh. 
    
    \begin{figure}[b]  
        \vspace{-10mm}
        \centering\includegraphics[height=4cm]{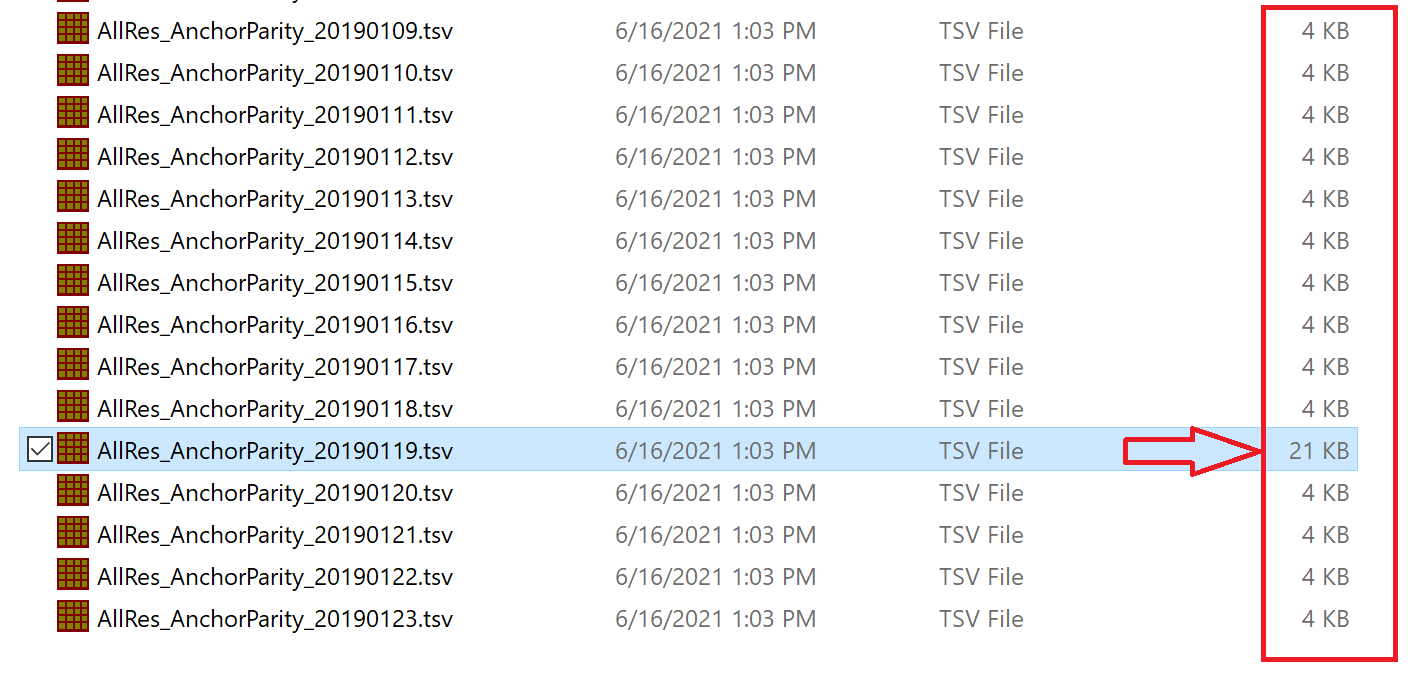}
        \vspace{-8mm}
        \caption{Example of a real DQ issue flagged by \avh.} 
        \label{fig:error-fp}
    \end{figure}

}
{

}

\end{document}